\documentclass[12pt, draftclsnofoot, onecolumn]{IEEEtran}
\usepackage{mathrsfs}
\usepackage{amsmath}
\usepackage{amssymb}
\usepackage{amsthm}
\usepackage{subfigure}
\usepackage{framed}
\usepackage{lineno}
\usepackage{multirow}
\usepackage{booktabs}
\usepackage{graphicx}
\usepackage{epstopdf}
\usepackage{algorithm}
\usepackage{algorithmic}
\usepackage{color}
\usepackage{multicol}
\usepackage{stfloats}
\usepackage{subfigure}
\usepackage{cite}
\usepackage{bm}
\usepackage{amsmath,amsfonts,amssymb}
\ifCLASSINFOpdf \else \fi

\DeclareMathOperator{\Tr}{Tr}
\DeclareMathOperator{\diag}{diag}

\newtheorem{Theo}{Theorem}

\newtheorem{Corol}{Corollary}
\begin{document}
	
\title{Waveform Optimization for MIMO Joint Communication and Radio Sensing Systems with Training Overhead\\}
\author{Xin Yuan,
	Zhiyong Feng,
	J. Andrew Zhang,
	Wei Ni,	\\
	Ren Ping Liu,
	Zhiqing Wei,
	and Changqiao Xu

	\thanks{ X. Yuan is with the Key Laboratory of the Universal Wireless
		Communications, Ministry of Education, Beijing University of Posts and Telecommunications, China, and the Global Big Data Technologies Center, University of Technology Sydney, Australia (email: yuanxin@bupt.edu.cn).}
	
	\thanks{ Z. Feng and Z. Wei are with the Key Laboratory of the Universal Wireless Communications, Ministry of Education, Beijing University of Posts and Telecommunications, China (email: \{fengzy, weizhiqing\}@bupt.edu.cn).}
	
	\thanks{ J. A. Zhang and R. P. Liu are with the Global Big Data Technologies Center, University of Technology Sydney, Australia (email: \{renping.liu, andrew.zhang\}@uts.edu.au).}
	
	\thanks{ W. Ni is with Commonwealth Scientific and Industrial Research Organization, Australia (email: wei.ni@data61.csiro.au).}
	
	\thanks{C. Xu is with the State Key Laboratory of Networking
		and Switching Technology, Beijing University of Posts and Telecommunications, China (e-mail: cqxu@bupt.edu.cn)}
}	
	
\maketitle

\begin{abstract}
In this paper, we study optimal waveform design to maximize mutual information (MI) for a joint communication and (radio) sensing (JCAS, a.k.a., radar-communication) multi-input multi-output (MIMO) downlink system. We consider a typical packet-based signal structure which includes training and data symbols. We first derive the conditional MI for both sensing and communication under correlated channels by considering the training overhead and channel estimation error (CEE). Then, we derive a lower bound for the channel estimation error and optimize the power allocation between the training and data symbols to minimize the CEE. Based on the optimal power allocation, we provide optimal waveform design methods for three scenarios, including maximizing MI for communication only and for sensing only, and maximizing a weighted sum MI for both communication and sensing. We also present extensive simulation results that provide insights on waveform design and validate the effectiveness of the proposed designs. 	

\end{abstract}

\begin{IEEEkeywords}
Mutual information, joint communication and sensing, waveform design, training sequence.
\end{IEEEkeywords}

%
\IEEEpeerreviewmaketitle

\section{Introduction}\label{sec::introduction}
\subsection{Background and Motivation}
A joint communication and (radio) sensing (JCAS, a.k.a., Radar-Communications) system that enables share of hardware and signal processing modules, can achieve efficient spectrum efficiency, enhanced security, and reduced cost, size, and weight~\cite{Kumari2018Sparsity,Chiriyath2017Radar,Paul2017Survey, Lushan2019}. 
JCAS systems can have many potential applications in intelligent transportation that require both communication links connecting vehicles and active environment sensing functions~\cite{Daniels18, zhang2019multibeam}. For JCAS systems, it is crucial to use a waveform simultaneously performing both communication and sensing function, and help improve the availability of the limited spectrum resources. To this end, one of the main challenges in JCAS systems lies in designing optimal or adequate waveforms that serve both purposes of data transmission and radio sensing.

Mutual information (MI) is an important measure that can be used for studying waveform designs for joint communication and sensing systems. To be specific, for communications the MI between wireless channels and the received communication signals can be employed as the waveform optimization criterion, while for sensing, the conditional MI between sensing channels and the reflected sensing signals can be measured~\cite{Bell1993Information,Yang2007MIMO}. 
Despite a significant amount of research effort on waveform design in both communication and sensing systems, existing joint waveform designs for JCAS systems are still limited. It is known that the training sequence for channel estimation has a significant impact on communication capacity, particularly for multiple input multiple output (MIMO) systems \cite{Hassibi2003How, Biguesh2006}. However, there has been no study on the waveform design for JCAS, which takes into consideration the typical signal packet structure containing the training sequence.

\subsection{Related Work}
Information theory has been used to design radar waveform~\cite{Bell1993Information,Tang2010MIMO,Yang2007MIMO,Zhu2017Information,Chen2013Adaptive}. Bell~\cite{Bell1993Information} was the first to apply information theory to optimize radar waveforms to improve target detection.
In~\cite{Zhu2017Information}, the optimal radar waveform was proposed to maximize the detection performance of an extended target in a colored noise environment by using MI as waveform design criteria. Two criteria, namely, the maximization of the conditional MI and the minimization of the minimum mean-square error (MMSE), were studied in~\cite{Yang2007MIMO} to optimize the waveform design for MIMO radars by exploiting the covariance matrix of the extended target impulse response.
In~\cite{Tang2010MIMO}, the optimal waveform design for MIMO radars in colored noise was also investigated by considering two criteria: maximizing the MI and maximizing the relative entropy between two hypotheses that the target exists or does not exist in the echoes. 
In~\cite{Chen2013Adaptive}, a two-stage waveform optimization algorithm was proposed for an adaptive MIMO radar to unify the signal design and selection procedures. The algorithm is based on the constant learning of the radar environment at the receivers and the adaptation of the transmit waveform to dynamic radar scene.
In~\cite{Liu2016Robust}, a robust waveform design based on the Cram{\'e}r-Rao bound was proposed for co-located MIMO radars to improve the worst-case estimation accuracy in the presence of clutters.

For communication and radar co-existing systems that transmit and process respective signals, the MI has also been adopted for waveform design to minimize the interference to each other.
In~\cite{Chiriyath2016Joint,Chiriyath2016Inner}, inner bounds on both the radar estimation rate for sensing and the data rate for communication were derived for the co-existing systems. 
Liu \textit{et al.}~\cite{Liu2017Robust} studied transmit beamforming for spectrum sharing between downlink MU-MIMO communication and co-located MIMO radar, to maximize the detection probability for sensing while guaranteeing the transmit power for downlink users.
In~\cite{Chiriyath2019Novel}, a minimum-estimation-error-variance waveform design method was proposed to optimize the spectral shape of a unimodular radar waveform and maximize the performance of both the radar and communications
In~\cite{Paul2016Joint}, the radar waveform was designed based on a performance bound that is derived from jointly maximizing radar estimation rate and communication data rate.

Only a few studies have investigated the MI for JCAS systems~\cite{XU2015102,Liu2017Adaptive,liu2018robust}.
In~\cite{XU2015102}, considering a JCAS MIMO setup, the expressions for radar mutual information and communication channel capacity were derived. In~\cite{Liu2017Adaptive,liu2018robust}, an integrated waveform design was proposed for OFDM JCAS systems to improve the MI for both communication and sensing by considering extended targets and frequency-selective fading channels. 

\subsection{Contributions}
This paper presents information theoretically optimal waveform designs for a JCAS MIMO downlink system with a signal packet structure, including training sequence and information data symbols. In the JCAS MIMO downlink, a node sends MIMO signals to another node for communications and simultaneously uses the reflected signals for sensing the surrounding environment. We first derive the conditional MI for sensing and communication by taking both the training overhead and channel estimation error (CEE) into consideration, and then provide the optimal waveform designs for several different options of maximizing conditional MI. For sensing, both training and data sequences directly contribute to the MI; while for communications, only the data sequence contributes to the MI in the presence of CEE linked to the training sequence.

The key contributions of the paper are summarized as follows.
\begin{enumerate}
	\item We derive the MI expressions for both sensing and communication. We reveal the significantly different impact of the training and data sequences on the MI of sensing and communications.
	\item We design the optimal power allocation scheme between the training and data sequences under MMSE estimators for correlated MIMO communication channels. 
	\item We provide the optimal waveform designs for three scenarios, including maximizing the MI only for sensing, and only for communication, and maximizing the weighted MI for joint communication and sensing. 
	\item We conduct extensive simulations to corroborate the effectiveness of our proposed power allocation and waveform design. The results provide important insights into the trade-off of MI between communication and sensing in a JCAS system, and the non-negligible impact of training sequence on the MI.	
\end{enumerate}

\subsection{Organization}
The rest of this paper is organized as follows. In Section~\ref{sec::system model}, the system model is introduced. In Section~\ref{sec::MI}, we derive the conditional MI for both communication and sensing in the JCAS system. In Section \ref{sec::cee}, we derive a lower bound for CEE and develop an optimal power allocation strategy between training and data sequences. In Section~\ref{sec::waveform design}, the optimal waveform design methods for optimal communication, optimal sensing, and JCAS are investigated. Section \ref{sec::simulation} presents simulation results. Section~\ref{sec::conclusion} concludes the paper.

\textit{Notation:} Lower-case bold face $(\mathbf{x})$ indicates vector, and upper-case bold face $(\mathbf{X})$ indicates matrix. For a diagonal matrix $\mathbf{X}$, $\mathbf{X}^a$ denotes the power $a$ operation to each diagonal element. $\mathbf{I}$ denotes the identity matrix, $\mathbb{E}(\cdot)$ denotes expectation. $(\cdot)^T$, $(\cdot)^H$, $(\cdot)^{\ast}$, $(\cdot)^{-1}$ and $(\cdot)^\dagger$ denote transposition, conjugate transportation, conjugate, inverse and pseudo-inverse, respectively. $\det(\cdot)$ and $\Tr(\cdot)$ denote the determinant and trace of a matrix, respectively.
	
\section{System Model}\label{sec::system model}
\begin{figure}
	\centering
	\includegraphics[width=0.45\textwidth]{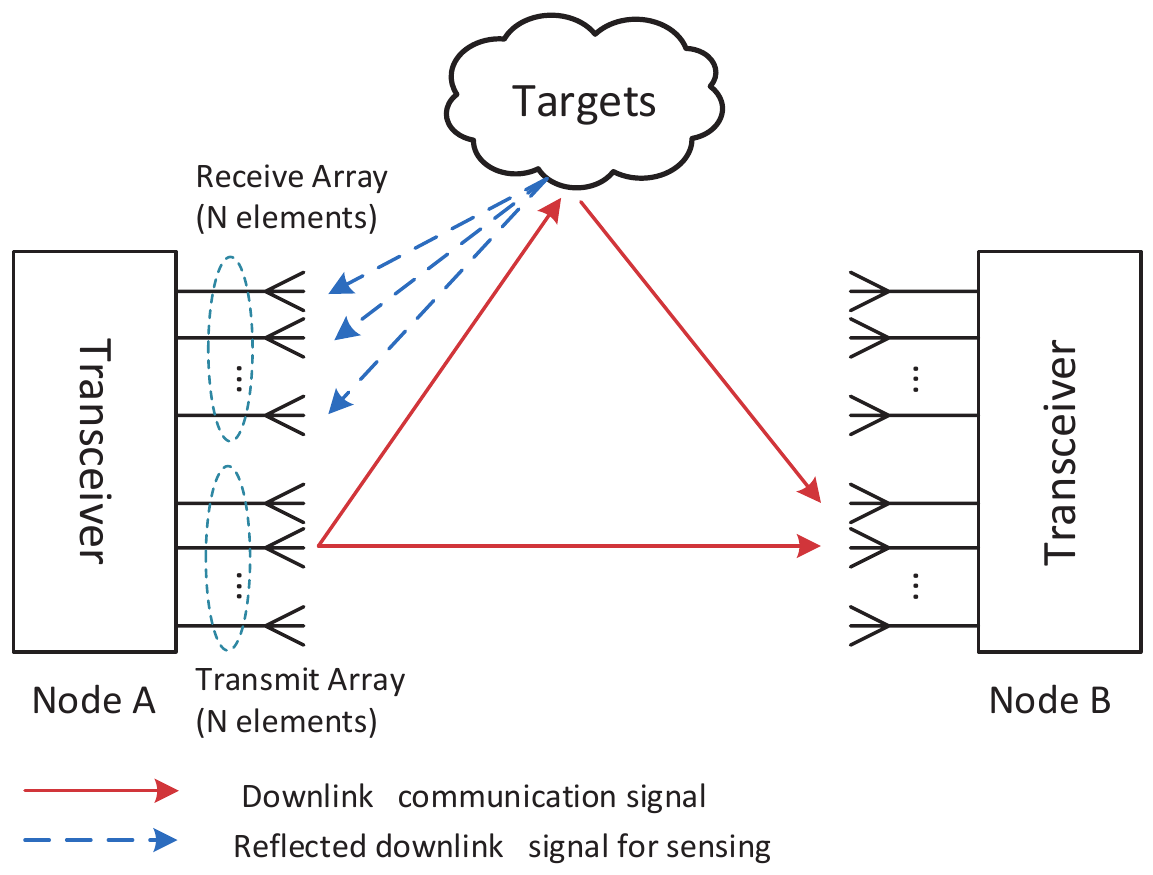}
	\caption{A joint communication and sensing (JCAS) MIMO downlink system, where node A transmits data to node B, and simultaneously senses the environment to determine, e.g., the locations and speeds of the nearby objects, by using the reflected transmitted signal.}
	\label{fig::system model}
\end{figure}

We consider a JCAS MIMO system where two nodes A and~B perform point-to-point communications in time division duplex (TDD) mode, and simultaneously sense the environment to determine, e.g., the locations and speeds of nearby objects, as illustrated in Fig.~\ref{fig::system model}. Each node has $N$ antennas configured in the form of a uniform linear array (ULA). At the stage that node A is transmitting to node B, we consider downlink sensing where the reflection of the transmitted signal is used for sensing by node~A. 
The transmitted symbols are known to node A. The channels of sensing and communications are correlated but different. 
To suppress leakage signals from the transmitter and enable the reception of clear sensing signals, each node is assumed to be equipped with two spatially widely separated antenna arrays, i.e., $N$ transmit antennas and $N$ receive antennas configured in the form of two uniform linear arrays (ULAs). Detailed configurations of the transceiver for JCAS systems are beyond the scope of this paper, and readers can refer to~\cite{Lushan2019}~and~\cite{zhang2019multibeam} for more details. 

\begin{figure}
	\centering
	\includegraphics[width=0.45\textwidth]{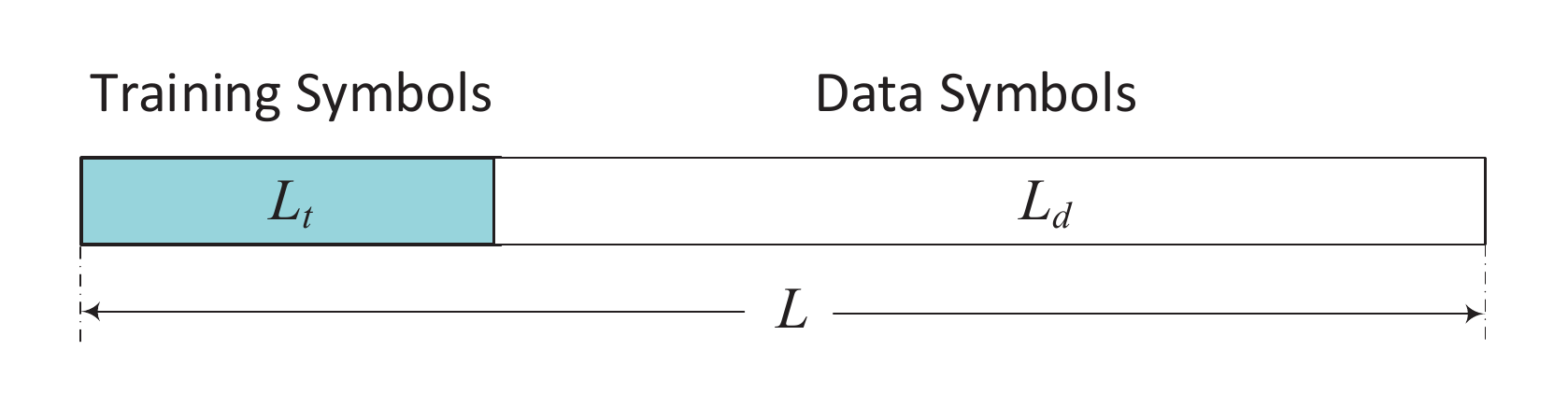}
	\caption{Transmit symbols: including training and data symbols. For communications, the non-precoded training symbols are used for synchronization and channel estimation, and the data symbols are typically precoded data payload. While for sensing, both the training and data symbols are used for targets detection}.
	\label{fig::training_signal_design}
\end{figure}

In practice, a communication packet typically includes data payload, together with training signals for synchronization and channel estimation. The training signals can have various forms in different standards and systems. For example, it can be comb pilots or occupy whole resource blocks in 5G New Radio. Without loss of generality, we consider a general data structure which consists of a sequence of $L_t$ training symbols and $L_d$ data symbols for each spatial stream, as illustrated in Fig.~\ref{fig::training_signal_design}. Concatenating the symbols from all $N$ spatial streams into a matrix $\mathbf{X}$, we have 
$\mathbf{X}=[\mathbf{X}_t,\mathbf{X}_d]$,
where $\mathbf{X}_t=\left[\bm{X}_{t}\left(1\right) ,\cdots,\bm{X}_{t}\left({N} \right) \right]^T \in {\mathbb{C}}^{N \times L_t}$ and $\mathbf{X}_d=\left[\bm{X}_d\left(1\right) ,\cdots,\bm{X}_d\left({N}\right) \right]^T \in {\mathbb{C}}^{N \times L_d}$, with $\bm{X}_t(n)$ and $\bm{X}_d(n)$ denoting the training and data symbols transmitted from the $n$-th antenna, respectively. We assume that $\mathbf{X}_d(n) \in \mathbb{C}^{L_d \times 1}$ is independent and identically distributed (i.i.d.) Gaussian variable with zero mean and covariance matrix $\frac{1}{L_d}\mathbb{E}\left\lbrace \mathbf{X}_d\mathbf{X}_d^H\right\rbrace = \varSigma_{\mathbf{X}_d}$. Let $ \frac{1}{L_t}\mathbf{X}_t\mathbf{X}_t^H = \varSigma_{\mathbf{X}_t}$. $\mathbf{X}_t(n) \in \mathbb{C}^{L_t \times 1}$ are typically designed to be orthogonal to each other and $L_t\geq N$, and hence $\varSigma_{\mathbf{X}_t}$ is a scaled diagonal matrix. More advanced designs of training sequences may be possible. The orthogonal design considered here is a typical setting in MIMO communication systems, and it is also typically used in MIMO radar to exploit the degrees of freedom offered by multiple antennas~\cite{Galati2013Waveforms}. 
Most of the results presented in this paper can also be readily extended to systems using other training sequences.
 
The transmitted signal $\mathbf{X}$, including $\mathbf{X}_t$ and $\mathbf{X}_d$, is used for both communication and radio sensing operations.
Let $P$ be the total energy of the transmit signal, $P_t$ the energy of the training signals, and $P_d$ the energy of the data signals. $P=P_t+P_d$. The average energy of the training and data symbols are $\sigma_t^2=\frac{1}{N L_t}\sum_{n=1}^{N} \mathbf{X}_t(n)^H\mathbf{X}_t(n)$, and $\sigma_d^2=\frac{1}{N L_d}\sum_{n=1}^{N} \mathbb{E}\left[\mathbf{X}_d(n)^H\mathbf{X}_d(n)\right]$, respectively. We also define a weighting value $\kappa$, $0<\kappa<1$, and have $P_d=\kappa P=NL_d\sigma_d^2$ and $P_t=(1-\kappa)P=\sigma_t^2NL_t$. We optimize the power allocation between training and data symbols to maximize the communication capacity, as will be described in Section~\ref{subsec::Optimal Training Signals Design}.

\subsection{Communication Model}\label{subsec::communication}
For communication, the received training and data signals at node B can be respectively given by
\begin{equation}\label{training_signal}
{\mathbf{Y}}^t_{\rm com} = \mathbf{H}\mathbf{X}_t +\mathbf{N}_{tc};
\end{equation}
\begin{equation}\label{communication_data_signal}
\begin{aligned}
{\mathbf{Y}}^d_{\rm com} & = \mathbf{H}\mathbf{X}_d +\mathbf{N}_{dc}\\
& = \left(\hat{\mathbf{H}} + \Delta \mathbf{H}\right)\mathbf{X}_d +\mathbf{N}_{dc}\\
& = \hat{\mathbf{H}}\mathbf{X}_d +\underbrace{\Delta \mathbf{H}\mathbf{X}_d +\mathbf{N}_{dc}}_{\mathbf{N}'_c},
\end{aligned}
\end{equation}
where $\mathbf{H}=\left[\mathbf{h}_1,\cdots,\mathbf{h}_j,\cdots, \mathbf{h}_{N}\right] \in \mathbb{C}^{N \times  N} $ is the channel matrix with $\mathbf{h}_j=\left[h_{1,j},h_{2,j},\cdots,h_{N,j}\right]^T$ denoting the $j$-th row of~$\mathbf{H}$; $\mathbf{N}_{tc} \in \mathbb{C}^{N \times L_t}$ and $\mathbf{N}_{dc} \in \mathbb{C}^{N \times L_d}$ are both addictive white Gaussian noise (AWGN) with zero mean and element-wise variance $\sigma_n^2$. 
It is reasonable to assume that $\mathbf{N}_{tc}$, $\mathbf{N}_{dc}$ and $\mathbf{X}_d$ are mutually independent. The signal ${\mathbf{Y}}^t_{\rm com}$ is used for channel estimation. We assume that a linear channel estimation based on a minimum mean-square error (MMSE) criterion~\cite{Artigue2011On} is applied. In this case, the channel estimate $\hat{\mathbf{H}}$ and the estimation error $\Delta \mathbf{H}$ are uncorrelated~\cite{Wang2009Near}. Let $\Delta \mathbf{H}=[\Delta\mathbf{h}_{1},\cdots,\Delta\mathbf{h}_{j},\cdots,\Delta\mathbf{h}_{N}]$, where $\Delta\mathbf{h}_{j} = [\Delta h_{1j},\Delta h_{2j},\cdots,\Delta h_{Nj}]^T$ the $j$-th row of~$\Delta\mathbf{H}$. The coefficients $\Delta h_{ij}$ are random variables following i.i.d. zero mean circularly symmetric complex Gaussian with variance $\sigma_e^2$, i.e., $\mathbb{E}\left[\Delta \mathbf{H}\Delta \mathbf{H}^H \right] = N \sigma_e^2\mathbf{I}_{N}$. We will evaluate $\sigma_e^2$ and link it to $\mathbf{X}_t$ and $\mathbf{N}_{tc}$ in Section \ref{sec-cee}.

The matrix $\mathbf{N}'_c$ combines the CEE and noise, and can be viewed as an equivalent additive noise with zero mean and covariance.
The variance $\sigma_{n'}^2$ can be obtained as 	
\begin{subequations}
\begin{align}
\mathbb{E}[\mathbf{N}'_c{\mathbf{N}'_c}^H] & = \mathbb{E}[\Delta \mathbf{H}\mathbf{X}_d\mathbf{X}_d^H\Delta \mathbf{H}^H]+\mathbb{E}\left[\mathbf{N}_{dc}\mathbf{N}_{dc}^H\right]\\
& = \mathbb{E}[\Delta \mathbf{H}\varSigma_{\mathbf{X}_d}\Delta \mathbf{H}^H]+\mathbb{E}\left[\mathbf{N}_{dc}\mathbf{N}_{dc}^H\right]\\
& = 
\mathbb{E}\left\lbrace \diag\left\lbrace \Delta\mathbf{h}_{1}^T \varSigma_{\mathbf{X}_d}\Delta\mathbf{h}_{1}^{\ast},\cdots,\Delta\mathbf{h}_{N}^T \varSigma_{\mathbf{X}_d}\Delta\mathbf{h}_{N}^{\ast}\right\rbrace \right\rbrace \!+\! L_d\sigma_n^2\mathbf{I}_N \\
& = \diag\left\lbrace {\Tr}\left(\varSigma_{\mathbf{X}_d}\mathbb{E}[\Delta\mathbf{h}_{1}^{\ast}\Delta\mathbf{h}_{1}^{T}]\right),\cdots,{\Tr}\left(\varSigma_{\mathbf{X}_d}\mathbb{E}[\Delta\mathbf{h}_{N}^{\ast}\Delta\mathbf{h}_{N}^T]\right)\right\rbrace + L_d\sigma_n^2\mathbf{I}_N \\
& = N L_d\sigma_d^2\sigma_e^2 \mathbf{I}_N +L_d\sigma_n^2\mathbf{I}_N
= L_d\left(\frac{P_d}{L_d}\sigma_e^2+\sigma_n^2\right) \mathbf{I}_N\\
&\triangleq L_d\sigma_n'^2\mathbf{I}_N,
\end{align}	
\end{subequations}
where $\sigma_n'^2 = \frac{P_d}{L_d}\sigma_e^2+\sigma_n^2$.

Let $\mathbf{R}_H = \frac{1}{N}\mathbb{E}[\mathbf{H}^H\mathbf{H}]$ be the channel covariance matrix, and $\mathbf{R}_H$ is a positive semi-definite matrix. We assume that $\mathbf{R}_H$ is known to Node A. We can write the random channel matrix as $\mathbf{H}=\mathbf{H}_0 \mathbf{R}_H^{\frac{1}{2}}$, where the entries of~$\mathbf{H}_0$ are i.i.d. zero mean circularly symmetric complex Gaussian with unit variance.  

\subsection{Sensing Model}\label{subsec::sensing}
Node A uses the reflection of the transmitted signal for sensing. The received signal, denoted by $\mathbf{Y}_{\rm rad}$, is given by 
\begin{equation}\label{eq::radar_signal_training}
\begin{aligned}
\mathbf{Y}_{\rm rad} & =  \mathbf{G}\mathbf{X}+\mathbf{N}
= \mathbf{G}[\mathbf{X}_t,\mathbf{X}_d] +\left[\mathbf{N}_{tr},\mathbf{N}_{dr}\right]\\
&= [\mathbf{G}\mathbf{X}_t+\mathbf{N}_{tr},\mathbf{G}\mathbf{X}_d+\mathbf{N}_{tr}], 
\end{aligned}
\end{equation}
where $\mathbf{G} = [\mathbf{g}_1,\cdots,\mathbf{g}_N]$ is the channel matrix to be sensed with its $j$-th column being $\mathbf{g}_j=\left[g_{1j},g_{2j},\cdots,g_{N j}\right]^T$, and $\mathbf{g}_j,\,j=1,\cdots,N$ are independent of each other; $\mathbf{N}_{tr} =[\mathbf{n}_{tr,1},\mathbf{n}_{tr,2},\cdots,\mathbf{n}_{tr,N}] \in \mathbb{C}^{L_t \times N}$ and $\mathbf{N}_{dr}=[\mathbf{n}_{dr,1},\mathbf{n}_{dr,2},\cdots,\mathbf{n}_{dr,N}] \in \mathbb{C}^{L_d \times N}$ are AWGN with zero mean and covariance matrix 
$\mathbb{E}\left\lbrace\mathbf{N}_{tr}{\mathbf{N}_{tr}^H}\right\rbrace = N\sigma_n^2 \mathbf{I}_{L_t}$ and $\mathbb{E}\left\lbrace\mathbf{N}_{dr} {\mathbf{N}_{dr}^H}\right\rbrace = N\sigma_n^2 \mathbf{I}_{L_d}$. Let $\varSigma_\mathbf{G}= \frac{1}{N}\mathbb{E}\{\mathbf{G}\mathbf{G}^H\}$ be the spatial correlation matrix. It is assumed to be full-rank and also known to Node A. 

For both the communication and sensing channels, we assume that they remain unchanged during the period of a packet. 
Note that for both communication and sensing, the channel matrices include large-scale path loss and small-scale fading. The path loss of sensing can vary significantly for different multi-path components depending on the number of nearby objects and their locations, and therefore, we consider the mean path loss herein. Our optimization results only depend on the ratio between the mean path losses of communication and sensing.



\section{Mutual Information}\label{sec::MI}
In this section, we first derive the expression for the MI of sensing by using both the training and data symbols. Then, we present the MI for communications under CEEs.

\subsection{MI for Sensing}
The MI between the sensing channel matrix $\mathbf{G}$ (or the ``target impulse response'' matrix in radar) and reflected signals $\mathbf{Y}_{\rm rad}$ given the knowledge of $\mathbf{X}$ can be used to measure the sensing performance~\cite{Tang2010MIMO}. With our model \eqref{eq::radar_signal_training}, the MI is given by
\begin{equation}\label{eq::MI_definition}
\begin{aligned}
I\left( \mathbf{G};{\mathbf{Y}}_{\rm rad}|\mathbf{X}\right) &= h\left({\mathbf{Y}}_{\rm rad}|{\mathbf{X}} \right) -h\left({\mathbf{Y}}_{\rm rad}|{\mathbf{X}},\mathbf{G} \right)\\
& = h\left({\mathbf{Y}}_{\rm rad}|[{\mathbf{X}_t},{\mathbf{X}_d}]^T\right)
-h\left({\mathbf{Y}}_{\rm rad}|[{\mathbf{X}_t},{\mathbf{X}_d}]^T,\mathbf{G}\right)\\
& = h\left({\mathbf{Y}}_{\rm rad}|[{\mathbf{X}_t},{\mathbf{X}_d}]^T\right)-h\left(\mathbf{N}_r \right),
\end{aligned}
\end{equation}
where $h(\cdot)$ denotes the entropy of a random variable. 
Provided the noise vector $\mathbf{N}_{r,j} = \begin{bmatrix}
	\mathbf{n}_{tr,j}\\ 
	\mathbf{n}_{dr,j}
	\end{bmatrix},\,j=1,\cdots,N$ are independent of each other, the conditional probability density function (PDF) of $\mathbf{Y}_{\rm rad}$ conditioned on $\mathbf{X}$ is given by 
\begin{subequations}\label{eq::PDF_of_Prad 1}
\begin{align}
 p\left(\mathbf{Y}_{\rm rad}|\mathbf{X}\right) & = p\left({\mathbf{Y}}_{\rm rad}|[\mathbf{X}_t,\mathbf{X}_d]^T\right) =\prod_{j=1}^{N}p\left(\mathbf{y}_{{\rm rad},j}|[\mathbf{X}_t,\mathbf{X}_d]^T\right) \label{eq::PDF_of_Prad 1b}\\
&=\prod_{j=1}^{N} \frac{1}{\pi^{L}\det\left([\mathbf{X}_t,\mathbf{X}_d]^T\varSigma_\mathbf{G}[\mathbf{X}_t,\mathbf{X}_d]^{\ast}+\sigma_n^2\mathbf{I}_{L}\right)}\nonumber\\
&\qquad\qquad\times \exp\left(-{\mathbf{y}}_{{\rm rad},j}^H \left([\mathbf{X}_t,\mathbf{X}_d]^T \varSigma_\mathbf{G} [\mathbf{X}_t,\mathbf{X}_d]^{\ast} + \sigma_n^2\mathbf{I}_{L}\!\right)^{-1} \mathbf{y}_{{\rm rad},j}\right) \label{eq::PDF_of_Prad 1c} \\
&= \frac{1}{\pi^{L N}\det^{N}\left([\mathbf{X}_t,\mathbf{X}_d]^T\varSigma_\mathbf{G} [\mathbf{X}_t,\mathbf{X}_d]^{\ast}+\sigma_n^2\mathbf{I}_{L}\right)}\nonumber\\
&\qquad\qquad\times \exp\left\{ -{\Tr}\left[\left([\mathbf{X}_t,\mathbf{X}_d]^T \varSigma_\mathbf{G} [\mathbf{X}_t,\mathbf{X}_d]^{\ast} \!+\! \sigma_n^2\mathbf{I}_{L}\right)^{-1} \mathbf{Y}_{\rm rad} {\mathbf{Y}^H_{\rm rad}}\right]\!\right\}\! \label{eq::PDF_of_Prad 1d},
\end{align} 
\end{subequations}
where~\eqref{eq::PDF_of_Prad 1c} is obtained based on  the PDF of circularly symmetric complex Gaussian distribution, and
\begin{subequations}\label{eq::covariance y}
\begin{align}
 \mathbb{E}\{\mathbf{y}_{{\rm rad},i}\mathbf{y}^H_{{\rm rad},i}\}
& = \mathbb{E}\!\left\{\!\begin{bmatrix}
\mathbf{X}^T_t\mathbf{g}_j\!+\!\mathbf{n}_{tr,j}\\ 
\mathbf{X}^T_d\mathbf{g}_j\!+\!\mathbf{n}_{dr,j}
\end{bmatrix}\left[\mathbf{g}_j^H\mathbf{X}_t^{\ast}+\mathbf{n}^H_{tr,j},\mathbf{g}_j^H\mathbf{X}^{\ast}_d+\mathbf{n}^H_{dr,j}\right]\right\}\!\label{eq::covariance ya}\\
& = [\mathbf{X}_t,\mathbf{X}_d]^T\mathbb{E}\{\mathbf{g}_j\mathbf{g}_j^H\}[\mathbf{X}_t,\mathbf{X}_d]^{\ast}+\mathbb{E}\{\diag\{\mathbf{n}_{tr,j}\mathbf{n}^H_{tr,j},\mathbf{n}_{dr,j}\mathbf{n}^H_{dr,j}\}\}\label{eq::covariance yb}\\
& = \frac{1}{N}[\mathbf{X}_t,\mathbf{X}_d]^T\varSigma_\mathbf{G} [\mathbf{X}_t,\mathbf{X}_d]^{\ast}+\sigma_n^2\mathbf{I}_{L},\label{eq::covariance yc}
\end{align}
\end{subequations}
where~\eqref{eq::covariance yb} is conditioned on $\mathbf{X}$, and $\mathbb{E}\{\mathbf{g}_j\mathbf{g}_j^H\} = \frac{1}{N}\mathbb{E}\{\mathbf{G}\mathbf{G}^H\}=\varSigma_\mathbf{G}$ in~\eqref{eq::covariance yc} since $\mathbf{g}_j,\,j=1,\cdots,N$ are independent of each other.

Based on~\eqref{eq::PDF_of_Prad 1}, 
the entropy of $\mathbf{Y}_{\rm rad}$ conditional on $\mathbf{X}$ can be obtained as 
\begin{subequations}\label{eq::entropy_training}
\begin{align}
 h\left( \mathbf{Y}_{\rm rad}|\mathbf{X}\right)
& = L N \log_2(\pi)+ L N \!+\! N \log_2\left[ \det\left([\mathbf{X}_t,\mathbf{X}_d]^T\varSigma_{\mathbf{G}} [\mathbf{X}_t,\mathbf{X}_d]^{\ast} +\sigma_n^2 \mathbf{I}_{L}\right) \right]\\
& =  L N \log_2(\pi)+ L N \!+\! N \log_2\left[ \det\left([\mathbf{X}_t,\mathbf{X}_d]^{\ast}[\mathbf{X}_t,\mathbf{X}_d]^T\varSigma_{\mathbf{G}} \!+\!\sigma_n^2 \mathbf{I}_{N}\right) \right] \label{eq::entropy_training b}\\
& =  L N \log_2(\pi)+ L N \!+\! N \log_2\!\left[\!(\sigma_n^2)^{{\color{blue}L-N}} \det\!\left(\!\mathbf{X}^{\ast}_t \mathbf{X}^T_t\varSigma_{\mathbf{G}}\!+\!\mathbf{X}^{\ast}_d\mathbf{X}_d^T\varSigma_{\mathbf{G}} \! +\!\sigma_n^2 \mathbf{I}_{N}\!\right) \!\right]\!,\label{eq::entropy_training c} 
\end{align}
\end{subequations}
where~\eqref{eq::entropy_training c} is based on the \textit{Sylvester's determinant} theorem~\cite{gilbert1991positive}, i.e., 
\begin{equation}
	\begin{aligned}
	&\det\left(\mathbf{A}_{M\times N}\mathbf{B}_{N\times M}+\sigma_n^2\mathbf{I}_{M}\right)
	 = (\sigma_n^2)^{{\color{blue}M-N}} \det\left(\mathbf{B}_{N\times M}\mathbf{A}_{M\times N}+\sigma_n^2\mathbf{I}_{N}\right).
	\end{aligned}
\end{equation}
The columns of the noise matrix $\mathbf{N}_{r}$ follow the i.i.d. multivariate complex Gaussian distribution with zero mean and covariance matrix $\sigma^2_{n}\mathbf{I}_N$, and the entropy of $\mathbf{N}_{r}$ is given by
\begin{equation}\label{eq::entropy_noise 1}
\begin{aligned}
h\left( \mathbf{N}_{r}\right)
&= L N \log_2(\pi)+ L N+ N \log_2\left[ \det\left(\sigma_n^2\mathbf{I}_{N}\right) \right]. 
\end{aligned}
\end{equation}
By substituting~\eqref{eq::entropy_training} and~\eqref{eq::entropy_noise 1} into~\eqref{eq::MI_definition}, the MI for sensing can be obtained as\\
\begin{equation}
	\begin{aligned}\label{eq::MI_sensing}
	I\left( \mathbf{G};{\mathbf{Y}}_{\rm rad}|\mathbf{X}\right) & \!=\! N \log_2\!\left[\! \det\left(\frac{\mathbf{X}^{\ast}_t \mathbf{X}^T_t\varSigma_{\mathbf{G}}\!+\!\mathbf{X}^{\ast}_d\mathbf{X}_d^T\varSigma_{\mathbf{G}}}{\!\left(\!\sigma_n^2\right) ^{{\color{blue}L-N}}}\!+\!\mathbf{I}_{N}\!\right)\!\!\right]\!.
		\end{aligned}
\end{equation}

\subsection{MI for Communication}
The MI for communication is defined as the mutual dependence between the transmit signals of node A and the received signals of node B, conditional on the estimated channel matrix $\hat{\mathbf{H}}$. With the Gaussian assumption of CEE, the conditional PDF of ${\mathbf{Y}}^d_{\rm com}$ on $\hat{\mathbf{H}}$ is given by
\begin{subequations}\label{PDF_of_Pcom}
\begin{align}
&p\left( \mathbf{Y}^d_{\rm com}|\hat{\mathbf{H}}\right)
=\prod_{i=1}^{L_d}p\left(\mathbf{y}^d_{{\rm com},i}|\hat{\mathbf{H}} \right)\label{PDF_of_Pcom a}\\
&=\prod_{i=1}^{L_d} \frac{1}{\pi^{N}\det\left(\hat{\mathbf{H}}\varSigma_{\mathbf{X}_d} {\hat{\mathbf{H}}}^H+\sigma_{n'}^2\mathbf{I}_{N}\right) } \exp\left( -{\mathbf{y}^d_{{\rm com},i}}^H \left(\hat{\mathbf{H}} \varSigma_{\mathbf{X}_d} {\hat{\mathbf{H}}}^H+ \sigma_{n'}^2\mathbf{I}_{N}\right)^{-1} \mathbf{y}^d_{{\rm com},i} \right) \label{PDF_of_Pcom b}\\
&= \frac{1}{\pi^{L_d N}\det^{L_d}\left(\hat{\mathbf{H}}\varSigma_{\mathbf{X}_d} {\hat{\mathbf{H}}}^H+\sigma_{n'}^2\mathbf{I}_{N}\right) } \exp\!\left\{\! -{\Tr}\!\left[\!\left(\hat{\mathbf{H}}\varSigma_{\mathbf{X}_d} {\hat{\mathbf{H}}}^H + \sigma_{n'}^2 \mathbf{I}_{N}\!\right)^{-1} \mathbf{Y}_{\rm com} \mathbf{Y}_{\rm com}^H\!\right] \!\right\}\!,\label{PDF_of_Pcom c}
\end{align} 
\end{subequations}
where~\eqref{PDF_of_Pcom a} is under the assumption that the columns of $\mathbf{Y}^d_{\rm com}$ (or $\mathbf{X}_d$) are i.i.d., and~\eqref{PDF_of_Pcom b} is from the PDF of circularly symmetric complex Gaussian distribution. The columns of equivalent noise matrix $\mathbf{N}'_c$ follow the i.i.d. multivariate complex Gaussian distribution with zero mean and covariance matrix $\sigma_n'^2\mathbf{I}_N$. By referring to~\eqref{eq::entropy_training} -- \eqref{eq::entropy_noise 1}, the entropy of $\mathbf{N}'_c$ can be given by
\begin{equation}\label{eq::entropy equivalent noise}
h\left( \mathbf{N}'_c\right)
\!=\! L_d N \log_2(\pi) \!+\! L_d N \!+\! L_d \log_2\!\left[\! \det\!\left(\sigma_{n'}^2\mathbf{I}_{N}\!\right)\!\right]\!.
\end{equation}
Therefore, the conditional MI between $\mathbf{X}_d$ and ${\mathbf{Y}}_{\rm com}^d$ is obtained as 
\begin{equation}\label{eq::MI communication}
\begin{aligned}
I\left( \mathbf{X}_d;{\mathbf{Y}}^d_{\rm com}|\hat{\mathbf{H}}\right)&=
h\left({\mathbf{Y}}^d_{\rm com}|{\hat{\mathbf{H}}} \right) -h\left({\mathbf{Y}}^d_{\rm com}|{\mathbf{X}_d},\hat{\mathbf{H}}\right) \\
& =h\left({\mathbf{Y}}^d_{\rm com}|{\hat{\mathbf{H}}} \right)  -h\left(\mathbf{N}'_c \right)
= L_d \log_2\!\left[\!\det\left(\frac{\hat{\mathbf{H}}\varSigma_{\mathbf{X}_d} {\hat{\mathbf{H}}}^H}{\sigma_{n'}^2}\!+\!\mathbf{I}_{N}\!\right) \!\right]\!.
\end{aligned}
\end{equation}

Compared to the conventional MI results without consideration of CEE~\cite{Cruz2010MIMO}, we can see that the CEE here contributes as
$\sigma_n'^2 = \frac{P_d}{L_d}\sigma_e^2+\sigma_n^2$.

\section{Channel Estimation Error and Optimal Power Allocation}\label{sec::cee}
In this section, we first derive a lower bound for CEE with the use of the training symbols. Based on this lower bound, we then propose an optimal scheme for allocating energy between training and data symbols, to maximize an upper bound of the MI for communications. We optimize the power allocation with respect to communication, as its impact on sensing performance is much weaker.

\subsection{Channel Estimation Error}\label{sec-cee}
With an MMSE MIMO channel estimation, the estimated MIMO channel matrix can be expressed as~\cite{Biguesh2006}
\begin{equation}
\begin{aligned}
\hat{\mathbf{H}} &= \mathbf{H}\mathbf{X}_t\mathbf{X}_t^H\left(\sigma_n^2\mathbf{I}_{N} +\mathbf{X}_t\mathbf{X}_t^H\right)^{-1} + \mathbf{N}_t\mathbf{X}_t^H\left(\sigma_n^2\mathbf{I}_{N} +\mathbf{X}_t\mathbf{X}_t^H\right)^{-1}
= \mathbf{H}- \Delta \mathbf{H},
\end{aligned}
\end{equation}
where, as can be recalled, $\mathbf{X}_t$ is an $N \times L_t$ training symbol matrix whose elements have the average energy $\sigma_t^2$. Take the singular value decomposition (SVD) of $\mathbf{R}_H$. $\mathbf{R}_H=\mathbf{U}_H\bm{\Lambda}_H\mathbf{U}_H^H$, where the singular value matrix $\bm{\Lambda}_H=\diag(\delta_1,\delta_2,\cdots,\delta_N)$, and $\frac{1}{N}\sum_{i=1}^{N}\delta_i =\frac{1}{N} {\Tr}(\mathbf{R}_H) \triangleq \sigma_h^2$. 
Let ${\bm{\Lambda}}_{\rm CRLB}$ be the Cram\'er-Rao lower bound (CRLB) of the channel matrix estimation~\cite{Berriche2004}. We have
\begin{equation}
\begin{aligned}
\mathbb{E}\left[\Delta \mathbf{H}\Delta \mathbf{H}^H \right]
& = \mathbb{E}\left[\left(\mathbf{H}-\hat{\mathbf{H}} \right)\left(\mathbf{H}-\hat{\mathbf{H}} \right)^H \right]\\
& \!\geq\!{\bm{\Lambda}}_{\rm CRLB} = \left(\frac{\varSigma_{\mathbf{X}_t}}{\sigma_n^2} \!s+\! \mathbf{R}_H^{-1} \right)^{-1}
 \!=\! \left(\frac{\mathbf{U}_H^H\varSigma_{\mathbf{X}_t}\mathbf{U}_H}{\sigma_n^2} +{\bm{\Lambda}_H}^{-1} \right)^{-1}\\
&\! =\! {\diag} \!\left(\!\frac{\sigma_n^2\delta_1}{\sigma_n^2\!+\!L_t \sigma_t^2\delta_1},\!\cdots\!,\frac{\sigma_n^2\delta_i}{\sigma_n^2\!+\!L_t \sigma_t^2\delta_i}, \!\cdots\!,\frac{\sigma_n^2\delta_N}{\sigma_n^2\!+\!L_t \sigma_t^2\delta_N}\!\right)\!,
\end{aligned}
\end{equation}
which is due to the fact that $\varSigma_{\mathbf{X}_t}= \mathbf{X}_t\mathbf{X}_t^H = L_t \sigma_t^2 \mathbf{I}_{N}$. 
Therefore, {a lower bound of MMSE of total channel estimates,} denoted by $\mathcal{C}_t$, can be represented as
\begin{equation}\label{eq::CRLB definition}
\mathcal{C}_t ={\Tr}({\bm{\Lambda}}_{\rm CRLB}) = \sum_{i=1}^{N} \frac{\sigma_n^2\delta_i}{\sigma_n^2+L_t \sigma_t^2\delta_i}.
\end{equation}
Here, $\mathcal{C}_t$ is a function of $\delta_i,\,i=1,\cdots,N$ with the constraint that $\sum_{i=1}^{N}\delta_i = {\Tr}(\mathbf{R}_H)$. Therefore, we can further obtain the lower bound of $\mathcal{C}_t$ by applying Lagrange multiplier method.
The Lagrangian function can be written as 
\begin{equation}
\begin{aligned}
\mathcal{L}\left({\bm{\Lambda}}_H\right) = \sum_{i=1}^{N} \frac{\sigma_n^2\delta_i}{\sigma_n^2+L_t \sigma_t^2\delta_i}
+\tau \left(\sum_{i=1}^{N}\delta_i - {\Tr}(\mathbf{R}_H) \right) ,
\end{aligned}
\end{equation}
where $\tau$ is the Lagrange multiplier. By solving $\frac{\partial \mathcal{L}\left({\bm{\Lambda}_H}\right)}{\partial \delta_i} =0$, we get
\begin{equation*}
	\frac{\sigma_n^4}{\left(\sigma_n^2+\delta_i L_t \sigma_t^2\right)^2 }+\tau=0,
\end{equation*}
which shows that the lower bound is achieved when $\delta_1=\cdots=\delta_i\cdots=\delta_N$ and $\delta_i = \frac{1}{N}{\Tr}(\mathbf{R}_H)=\frac{1}{N}\sum_{i=1}^{N}\delta_i { = \sigma_h^2},\,i=1,\cdots,N$. The lower bound of $\mathcal{C}_t$ is then given by
\begin{equation}\label{eq::CEE LB}
	\mathcal{C}_t \geq \frac{N\sigma_n^2{\sigma_h^2}}{\sigma_n^2+L_t \sigma_t^2{\sigma_h^2}}.
\end{equation} 
Therefore, for any diagonal element of $\bm{\Lambda}_{\rm CRLB} (i)$, we have
\begin{equation}
\label{eq-ce}
{\bm{\Lambda}}_{\rm CRLB} (i)
  \geq \frac{\sigma_n^2{\sigma_h^2}}{\sigma_n^2+L_t \sigma_t^2{\sigma_h^2}} \triangleq \mathcal{C}_e, \; i=1,\cdots,N.
\end{equation}

\subsection{Optimal Power Allocation of Training and Data Symbols}\label{subsec::Optimal Training Signals Design}
In general, there are some constraints on the maximum and average transmission powers of a transmitter. When such power constraints are applied, there is a motivation for optimizing the power allocation between the training and data symbols, especially for maximizing the MI for communications.
Here, we optimize power allocation only by referring to the communication MI, because its impact on communication MI is much stronger than on sensing MI. Larger CEE can cause substantially deteriorate communication performance while sensing MI can only be slightly affected since the training sequence is directly used for sensing.

Since $\hat{\mathbf{H}} = \mathbf{H} - \Delta \mathbf{H}$, we can obtain that $\hat{\mathbf{H}}$ is a random variable with zero mean and variance $\sigma^2_{\hat{\mathbf{H}}} = \frac{1}{N^2} \mathbb{E}\left[{\Tr}\{\hat{\mathbf{H}}{\hat{\mathbf{H}}}^H\}\right]$. According to the orthogonality principle for MMSE~\cite{Hassibi2003How} and the obtained lower bound of CEE, we have $\sigma^2_{\hat{\mathbf{H}}} = \sigma_h^2 - \sigma_e^2$. 
Therefore, the estimated channel $\hat{\mathbf{H}}$ can be normalized as $\tilde{\mathbf{H}} = \frac{1}{\sigma_{\hat{\mathbf{H}}}}\hat{\mathbf{H}}$, which has elements following i.i.d. Complex Gaussian distribution $\mathcal{CN}(0,1)$.

The MI in~\eqref{eq::MI communication} can then be rewritten as
\begin{subequations}\label{eq::MI communication nor}
	\begin{align}
	 I\!\left(\! \mathbf{X}_d;{\mathbf{Y}}^d_{\rm com}|\hat{\mathbf{H}}\!\right) &= L_d \log_2\!\left[\! \det\!\left(\!\frac{\sigma^2_{\hat{\mathbf{H}}}}{\sigma_{n'}^2}\tilde{\mathbf{H}}\varSigma_{\mathbf{X}_d} {\tilde{\mathbf{H}}}^H \!+\!\mathbf{I}_{N}\!\right)\! \!\right]\!\label{eq::MI communication nor a}\\
	& = L_d \log_2\left[ \det\left(\frac{\sigma_h^2 - \sigma_e^2}{\frac{P_d}{ L_d}\sigma_e^2+\sigma_{n}^2}\tilde{\mathbf{H}}\varSigma_{\mathbf{X}_d} {\tilde{\mathbf{H}}}^H+\mathbf{I}_{N}\right) \right]\label{eq::MI communication nor b}\\
	& \leq  L_d \log_2\left[ \det\left(\frac{\sigma_h^2 - \mathcal{C}_e}{\frac{P_d}{ L_d}\mathcal{C}_e+\sigma_{n}^2}\tilde{\mathbf{H}}\varSigma_{\mathbf{X}_d} {\tilde{\mathbf{H}}}^H+\mathbf{I}_{N}\right) \right],\label{eq::MI communication nor c}
	\end{align}	
\end{subequations}
where~\eqref{eq::MI communication nor c} is due to the lower bound of CEE, i.e., $\sigma_e^2 \geq \mathcal{C}^l_t$.

By substituting $\mathcal{C}_e$ of \eqref{eq-ce} into~\eqref{eq::MI communication nor} and exploiting the Jensen's inequality, we can obtain the results for optimal power allocation and minimized CEE as summarized in the following theorem.

\begin{Theo}[Optimal Power Allocation]\label{theo::op ene_allo}
The optimal power allocation for maximizing the channel capacity under the training symbols is given by
	\begin{equation}\label{eq::k op}
		\kappa_{\rm{op}} = \left\{
		\begin{aligned}
		&\Gamma+\sqrt{\Gamma(\Gamma-1)},\; &{\rm if}\; L_d < N;\\
		&\frac{1}{2},\; &{\rm if}\; L_d = N;\\
		&\Gamma-\sqrt{\Gamma(\Gamma-1)},\; &{\rm if}\; L_d > N,
		\end{aligned}
		\right.
	\end{equation}
where $\Gamma = \frac{L_d}{L_d-N}\left(1+\frac{N \sigma_n^2}{P\sigma_h^2}\right)$.

The lower bound of the CEE can be given by
\begin{equation}\label{eq::CRLB}
\mathcal{C}_t^{l} = \frac{N\sigma_n^2{\sigma_h^2}}{ N\sigma_n^2+\left(1-\kappa_{\rm{op}}\right)P{\sigma_h^2}}.
\end{equation}
\end{Theo}
\begin{proof}
	The proof is provided in Appendix~\ref{proof::theo 1}.
\end{proof}

\begin{Corol}
	In a high SNR regime, $\Gamma$ and $\kappa_{\rm{op}}$ can be approximated as	
	$$\Gamma \approx \frac{L_d}{L_d-N};\;\kappa_{\rm{op}} \approx \frac{\sqrt{L_d}}{\sqrt{L_d}+\sqrt{N}}.$$
	In this case, the SNR is approximately $\rho_{\max}=\frac{{L_d}}{(\sqrt{L_d}+\sqrt{N})^2}\frac{P}{N\sigma_n^2}$.
	We can find that, in the high SNR regime, the power allocation between the data and the training symbols depends on the number of the data symbols, $L_d$, and the number of antennas, $N$. Moreover, $\kappa_{\rm{op}}$ decreases with the growth of both $L_d$ and $N$.
	
	In a low SNR regime, $\Gamma$ and $\kappa_{\rm{op}}$ can be approximated as
	$$\Gamma \approx \frac{L_dN\sigma_n^2}{(L_d-N)P\sigma_h^2};\;\kappa_{\rm{op}} \approx \frac{1}{2}.$$
	In this case, the SNR is approximately $\rho_{\max}=\frac{P^2 \sigma_h^4}{4N^2\sigma_n^4}$. We find that half of the total energy are allocated to the training symbols, and the maximum SNR in the low SNR regime quadratically increases with $\frac{P}{\sigma_{n}^2}$.
\end{Corol}

Hereafter, we use the optimal power allocation formula~\eqref{eq::k op} for power allocation and let $\mathcal{C}^{l}_e= \frac{\mathcal{C}_t^{l}}{N} = \frac{\sigma_n^2{\sigma_h^2}}{ N\sigma_n^2+\left(1-\kappa_{\rm{op}}\right)P{\sigma_h^2}}$; unless otherwise stated.

\section{Optimal Waveform Design} \label{sec::waveform design}
With the optimized power allocation in Section~\ref{sec::cee}, we investigate the waveform design for three scenarios in this section, which are maximizing the MI for sensing only and for communications only, and maximizing a weighted relative MI jointly for communication and radio sensing.

\subsection{Optimal Waveform Design for Radio Sensing Only}\label{sec::maximum MI}
In order to achieve the maximum MI for sensing, or in other words, to make received signals~${\mathbf{Y}}_{\rm rad}$ (including ${\mathbf{Y}}^t_{\rm rad}$ and ${\mathbf{Y}}^d_{\rm rad}$) containing rich information about $\mathbf{G}$, the transmit signals~$\mathbf{X}$ (including the training sequence~$\mathbf{X}_t$ and data sequence~$\mathbf{X}_d$) should be designed according to the sensing channel matrix $\mathbf{G}$. 
Since the training sequence $\mathbf{X}_t$ and the data sequence $\mathbf{X}_d$ are independent and have different correlations, 
the optimization problem for maximizing the MI for sensing can be decoupled into two separate optimization problems. 
As assumed, $\mathbf{X}_t$ contains deterministic orthogonal rows and {$ \mathbf{X}_t\mathbf{X}_t^H = L_t \sigma_t^2 \mathbf{I}_{N}$}. We only need to consider the optimization problem for the data sequence.

For the data sequence, the spatial correlation matrix can be diagonalized through SVD, that is,
\begin{equation}\label{eq::Sigma_G}
\varSigma_{\mathbf{G}}=\frac{1}{N}\mathbb{E}\{\mathbf{G}\mathbf{G}^H\}=\mathbf{U}_G\mathbf{\Lambda}_G\mathbf{U}_G^H,
\end{equation}
where $\mathbf{U}_G$ is a unitary matrix and 
$\mathbf{\Lambda}_G={\diag}\left\{ \lambda_{11},\cdots,\lambda_{ii},\cdots,\lambda_{NN}\right\}$
is a diagonal matrix with~$\lambda_{ii}$ being the singular values. 
The mean channel gain $\sigma_g^2$ for sensing channels is $\sigma_g^2= \sum_{i=1} \lambda_{ii}$.
The MI in~\eqref{eq::MI_sensing} can be rewritten as~\eqref{eq::MI_sensing2}, 
\begin{subequations}\label{eq::MI_sensing2}
		\begin{align}
		I\left( \mathbf{G};{\mathbf{Y}}_{\rm rad}|\mathbf{X}\right)	
		& = N \log_2\left[ \det\left(\frac{\mathbf{X}^{\ast}_t\mathbf{X}_t^T\mathbf{U}_G\mathbf{\Lambda}_G\mathbf{U}_G^H+\mathbf{X}^{\ast}_d\mathbf{X}_d^T\mathbf{U}_G\mathbf{\Lambda}_G\mathbf{U}_G^H }{\left(\sigma_n^2\right) ^{\frac{L}{N}}}+\mathbf{I}_{N}\right)\right]  \label{eq::MI_sensing2 a}\\
		& = N \log_2\left[ \det\left(\frac{\mathbf{\Lambda}_G\left(\mathbf{X}^T_t\mathbf{U}_G \right)^H \mathbf{X}^T_t\mathbf{U}_G\!+\!\mathbf{\Lambda}_G\left(\mathbf{X}^T_d\mathbf{U}_G \right)^H \mathbf{X}^T_d\mathbf{U}_G}{\left(\sigma_n^2\right) ^{\frac{L}{N}}}+\mathbf{I}_{N}\right)\!\right]\!,\label{eq::MI_sensing2 b}
		\end{align}
\end{subequations}
where \eqref{eq::MI_sensing2 b} is based on \textit{Sylvester's determinant} theorem.
Define $\mathbf{Q}^{(t)}=\left(\mathbf{X}^T_t\mathbf{U}_G \right)^H \mathbf{X}^T_t\mathbf{U}_G$ and $\mathbf{Q}^{(d)}=\left(\mathbf{X}^T_d\mathbf{U}_G \right)^H \mathbf{X}^T_d\mathbf{U}_G$, and their $(i,j)$-th entries are $q^{(t)}_{ij}$ and $q^{(d)}_{ij}$, respectively. 

According to Hadamard's inequality for the determinant and trace of an $N \times N$ positive semi-definite Hermitian matrix, we have the following inequalities: 
\begin{equation*}\label{eq::Hadamard inequality 1}
\det\left(\mathbf{Q}^{(t)}_{N \times N} \right) \leq \prod_{i=1}^{N} q^{(t)}_{ii};\;
\det\left(\mathbf{Q}^{(d)}_{N \times N} \right) \leq \prod_{i=1}^{N} q^{(d)}_{ii},
\end{equation*}
and
\begin{equation*}\label{eq::Hadamard inequality 3}
\Tr\left({\mathbf{Q}^{(t)}_{N \times N}}^{-1} \right) \geq \prod_{i=1}^{N}\frac{1}{q^{(t)}_{ii}};\;
\Tr\left({\mathbf{Q}^{(d)}_{N \times N}}^{-1} \right) \geq \prod_{i=1}^{N}\frac{1}{q^{(d)}_{ii}},
\end{equation*}
where the equalities are achieved if and only if $\mathbf{Q}_{N \times N}$ is diagonal.
As a result, the MI can be rewritten as
\begin{equation}\label{eq::MI_sensing3}
I\left( \mathbf{G};{\mathbf{Y}}_{\rm rad}|\mathbf{X}\right) \!\leq\! N \log_2\left[ \prod_{i=1}^{N}\left(\frac{\lambda_{ii}(q^{(t)}_{ii}\!+\!q^{(d)}_{ii})}{\left(\sigma_n^2\right) ^{\frac{L}{N}}}\!+\!1\right) \right].
\end{equation}

Since $\mathbf{U}_G$ is unitary, we can find $\Tr\left( \mathbf{Q}^{(d)}\right)=\Tr\left(\left(\mathbf{X}^T_d\mathbf{U}_G \right)^H \mathbf{X}^T_d\mathbf{U}_G \right)= \Tr\left(\mathbf{Z}^H\mathbf{Z} \right)$, where $\mathbf{Z}=\mathbf{X}^T_d\mathbf{U}_G$ is an $L_d \times N$ matrix. 
Under the constraint that the total transmit power is finite, we have $\Tr\left( \mathbf{Q}^{(d)}\right)  \leq P_d=\kappa_{\rm{op}}{P}$, where $P_d$ is the total power of the transmit data signals. Since $\mathbf{X}_t$ satisfies the orthogonality condition, we have $\Tr\left( \mathbf{Q}^{(t)}\right)=\Tr\left(\left(\mathbf{X}^T_t\mathbf{U}_G \right)^H \mathbf{X}^T_t\mathbf{U}_G \right)= \Tr\left(\mathbf{X}^{\ast}_t \mathbf{X}^t_t \right)= \Tr\left(\mathbf{X}_t \mathbf{X}^H_t \right)=P_t=(1-\kappa_{\rm{op}})P$ and $q^{(t)}_{ii}=\frac{P_t}{N}=\frac{(1-\kappa_{\rm{op}})P}{N}$.
Therefore, the maximum MI can be obtained by solving the following constrained problem:
\begin{equation}\label{eq::MI maximum}
\begin{aligned}
F_r 
&\!=\! \underset{\mathbf{Q}^{(d)}}{\max} \sum_{i=1}^{N} \left\lbrace\log_2\left(\frac{\lambda_{ii} \left(\frac{P_t}{N}+q^{(d)}_{ii}\right) }{\left(\sigma_n^2\right) ^{\frac{L}{N}}}+1\right)\right\rbrace,\\
\quad&\text{subject to} \;\; \Tr\left( \mathbf{Q}^{(d)}\right)  \leq P_d;\\
&\qquad\quad\quad q^{(d)}_{ii}\geq 0,\;1\leq i \leq N.
\end{aligned}
\end{equation}
We can apply the Lagrange multiplier method to solve \eqref{eq::MI maximum}. The Lagrangian can be written as
\begin{equation}
\mathcal{L}\!\left(\mathbf{Q}^{(d)}\!\right) \!=\!\sum_{i=1}^{N} \!\left\lbrace\!\log_2\!\left(\!\frac{\lambda_{ii} \!\left(\!\frac{P_t}{N}\!+\!q^{(d)}_{ii}\!\right)\! }{\left(\sigma_n^2\right) ^{\frac{L}{N}}}\!+\!1\!\right)\!\!\right\rbrace
\!+\alpha \sum_{i=1}^{N}
q^{(d)}_{ii},
\end{equation}
where $\alpha$ is the Lagrange multiplier associated with $q^{(d)}_{ii},\,i=1,\cdots,N$.
Differentiating $\mathcal{L}(\mathbf{Q}^{(d)})$ with respect to $q^{(d)}_{ii}$, $i=1,\cdots,N$, and setting the first-order derivative as 0, we can obtain $q^{(d)}_{ii}$ as follows:
\begin{equation}
\begin{aligned}
q^{(d)}_{ii}
&=-\frac{\left(\sigma_n^2\right) ^{\frac{L}{N}}}{\alpha\ln2}-\frac{P_t}{N}-\frac{\left(\sigma_n^2\right) ^{\frac{L}{N}}}{\lambda_{ii}}.
\end{aligned}
\end{equation}
The optimality conditions are satisfied if
\begin{equation}
\sum_{i=1}^{N}\left( -\frac{\left(\sigma_n^2\right) ^{\frac{L}{N}}}{\alpha\ln2}-\frac{P_t}{N}-\frac{\left(\sigma_n^2\right) ^{\frac{L}{N}}}{\lambda_{ii}}\right)^{+} = \kappa_{\rm{op}}{P}
\end{equation}
holds, and
\begin{equation}\label{eq::op power radar MI}
q^{(d)}_{ii} =\left( -\frac{\left(\sigma_n^2\right) ^{\frac{L}{N}}}{\alpha\ln2}-\frac{P_t}{N}-\frac{\left(\sigma_n^2\right) ^{\frac{L}{N}}}{\lambda_{ii}}\right)^{+}, i=1,\cdots,N.
\end{equation}

Since the diagonal elements of $\mathbf{Q}^{(d)}$ are real and greater than~$0$, ${\mathbf{Q}^{(d)}}^{\frac{1}{2}}$ exists. For any $L_d \times N$ matrix $\mathbf{\Psi}$ with orthogonal columns, if $\mathbf{\Psi}^H\mathbf{\Psi}=\mathbf{I}_{N}$, we have $\mathbf{Z}=\mathbf{\Psi}{\mathbf{Q}^{(d)}}^{\frac{1}{2}}$.
Since $\mathbf{Z}=\mathbf{X}^H_d\mathbf{U}_G$, $\mathbf{X}_d$ can be obtained as 
\begin{equation}\label{eq-xdsen}
\mathbf{X}_d=\left(\mathbf{\Psi}{\mathbf{Q}^{(d)}}^{\frac{1}{2}}\mathbf{U}_G^H\right)^H.
\end{equation}

With the obtained optimal $\mathbf{X}_d$ for sensing, we can derive the corresponding communication MI which is not necessarily optimal, as given by 
\begin{subequations}\label{eq::op capcity with maximum radar MI}
\begin{align}
I\left(\! \mathbf{X}_d;{\mathbf{Y}}^d_{\rm com}|\hat{\mathbf{H}}\!\right) &\leq  
L_d \log_2\!\left[\!\det\!\left(\!\frac{(\sigma_h^2 \!-\! \mathcal{C}^{l}_e)\tilde{\mathbf{H}}\varSigma_{\mathbf{X}_d} {\tilde{\mathbf{H}}}^H}{\frac{P_d}{ L_d}\mathcal{C}^{l}_e+\sigma_{n}^2}\!+\!\mathbf{I}_{N}\!\right)\! \!\right]\!\\
&\!=\! L_d \log_2\!\left[\! \det\!\left(\!\frac{(\sigma_h^2 - \mathcal{C}^{l}_e){\mathbf{U}_G \mathbf{Q}^{(d)}\mathbf{U}_G^H \tilde{\mathbf{H}}^H\tilde{\mathbf{H}}}}{\frac{P_d}{ L_d}\mathcal{C}^{l}_e+\sigma_{n}^2}\!+\!\mathbf{I}_{N}\!\right)\! \!\right]\!,\label{eq::op capcity with maximum radar MI b}
\end{align}
\end{subequations}
where~\eqref{eq::op capcity with maximum radar MI b} is obtained due to $\varSigma_{\mathbf{X}_d}=\frac{1}{L_d}\mathbb{E}\left\{\mathbf{X}_d\mathbf{X}_d^H\right\}=\left( \mathbf{\Psi}{\mathbf{Q}^{(d)}}^{\frac{1}{2}}\mathbf{U}_G^H\right)^H\mathbf{\Psi}{\mathbf{Q}^{(d)}}^{\frac{1}{2}}\mathbf{U}_G^H= \mathbf{U}_G \mathbf{Q}^{(d)}\mathbf{U}_G^H$.

\subsection{Optimal Waveform Design for Communication Only}\label{sec::maximum channel capacity}

After obtaining the optimal power allocation between the training and data symbols for maximizing the channel capacity in the presence of CEE, as shown in Theorem~\ref{theo::op ene_allo}, we can design waveform based on the optimal power allocation and correlated channel matrix $\mathbf{H}$ to maximize the MI  for communications (i.e., channel capacity). 
As derived in Appendix~\ref{proof::theo 1}, the MI maximization problem can be formulated to maximize the upper bound of the MI, as given by 
\begin{equation}\label{eq::op_communication with ce}
\begin{aligned}
F_c
&=\max L_d \sum_{i=1}^{N} \left\lbrace\log_2\left[\frac{\left(\sigma_h^2 - \mathcal{C}^{l}_e\right)\mu_{ii}\xi_{ii} }{\frac{P_d}{ L_d}\mathcal{C}^{l}_e+\sigma_n^2}+1\right] \right\rbrace\\
\quad&\text{subject to} \;\; \Tr\left( \mathbf{\Xi}\right) \leq \kappa_{\rm{op}}{P};\\
&\qquad\quad\quad\xi_{ii}\geq 0,\;1\leq i \leq N.
\end{aligned}
\end{equation}

The optimal solutions for this problem are satisfied if
\begin{equation}
\sum_{i=1}^{N}\left[ -\frac{1}{\beta'\ln2}-\frac{\sigma_n^2+\frac{P_d}{ L_d}\mathcal{C}^{l}_e}{ \mu_{ii}\left(\sigma_h^2 - \mathcal{C}^{l}_e\right)}\right]^{+} = \kappa_{\rm{op}}{P}
\end{equation}
holds, where $\beta'$ is the the Lagrange multiplier. The optimal singular values $\xi_{ii}$ of the correlation matrix $\mathbf{\Xi}$ can be obtained as 
\begin{equation}\label{eq::op power channel capacity}
\xi_{ii}=\left[ -\frac{1}{\beta'\ln2}-\frac{\sigma_n^2+\frac{P_d}{ L_d}\mathcal{C}^{l}_e}{ \mu_{ii}\left(\sigma_h^2 - \mathcal{C}^{l}_e\right)}\right]^{+}, \; i=1,\cdots,N.
\end{equation}

Since the diagonal elements of $\mathbf{\Xi}$ are real and positive, ${\mathbf{\Xi}}^{\frac{1}{2}}$ exists. For any $L_d\times N$ matrix $\mathbf{\Theta}$ with orthonormal columns, if $\mathbf{\Theta}^H\mathbf{\Theta}=\mathbf{I}_{N}$, we have $\mathbf{Y}^d_{\rm com}=\mathbf{\Theta}{\mathbf{\Xi}}^{\frac{1}{2}}$.
Since $\mathbf{Y}^d_{\rm com}=\mathbf{X}_d\mathbf{U}_{\tilde{\mathbf{H}}}$, $\mathbf{X}_d$ can be obtained as 
\begin{equation}
\mathbf{X}_d=\left(\mathbf{\Theta}{\mathbf{\Xi}}^{\frac{1}{2}}\mathbf{U}_{\tilde{\mathbf{H}}}^H\right)^H.
\end{equation}

Based on the optimal singular values of the covariance matrix for the data signals in~\eqref{eq::op power channel capacity}, we can obtain the MI for sensing under the condition of the maximum MI for communications, as given by
\begin{equation}\label{eq::op radar MI with maximum capcity}
	\begin{aligned}
	&I\left( \mathbf{G};{\mathbf{Y}}_{\rm rad}|\mathbf{X}\right)
	 = N \log_2\!\left[ \det\left(\!\frac{\left(\mathbf{X}^T_t\mathbf{U}_G \right)^H \mathbf{X}^T_t\mathbf{U}_G\!+\!\mathbf{U}_{\tilde{\mathbf{H}}} \mathbf{\Xi}\mathbf{U}_{\tilde{\mathbf{H}}}^H \varSigma_\mathbf{G}}{\left(\sigma_n^2\right) ^{\frac{L}{N}}}\!+\!\mathbf{I}_{N}\!\right)\!\right]\!,
	\end{aligned}
\end{equation}
\eqref{eq::op radar MI with maximum capcity} is obtained due to $\varSigma_{\mathbf{X}_d}=\frac{1}{L_d}\mathbb{E}\left\{\mathbf{X}^{\ast}_d\mathbf{X}_d^T\right\}=\frac{1}{L_d}\mathbb{E}\left\{\mathbf{X}_d\mathbf{X}_d^H\right\}=\left( \mathbf{\Theta}{\mathbf{\Xi}}^{\frac{1}{2}}\mathbf{U}_{\tilde{\mathbf{H}}}^H\right)^H\mathbf{\Theta}{\mathbf{\Xi}}^{\frac{1}{2}}\mathbf{U}_{\tilde{\mathbf{H}}}^H= \mathbf{U}_{\tilde{\mathbf{H}}} \mathbf{\Xi}\mathbf{U}_{\tilde{\mathbf{H}}}^H$.

\subsection{Joint Maximization of a Weighted Sum of MI}
In this section, we conduct the waveform optimization for jointly considering the MI for both communication and radio sensing. Since there is generally no solution that can simultaneously maximize the MI for communication and radio sensing, a weighted sum of them is exploited and given by
\begin{equation}\label{eq::weighted cost function 1}
	 F_w \!=\! \frac{w_r}{F_r}I\left( \mathbf{G};{\mathbf{Y}}_{\rm rad}|\mathbf{X}\right) + \frac{1-w_r}{F_c}I\left( \mathbf{X}_d;{\mathbf{Y}}^d_{\rm com}|\hat{\mathbf{H}}\right).	
\end{equation}

To maximize the weighted sum, the transmitted data signals $\mathbf{X}_d$ should be designed according to the correlation matrices of both $\mathbf{H}$ and $\mathbf{G}$. Based on the SVD, $\varSigma_{\mathbf{H}}=\frac{1}{N}\mathbb{E}\{\mathbf{H}\mathbf{H}^H\}=\mathbf{U}_H\mathbf{\Lambda}_H\mathbf{U}_H^H$ and $\varSigma_{\mathbf{G}}=\frac{1}{N}\mathbb{E}\{\mathbf{G}\mathbf{G}^H\}=\mathbf{U}_G\mathbf{\Lambda}_G\mathbf{U}_G^H$,~\eqref{eq::weighted cost function 1} can be rewritten as
\begin{subequations}\label{eq::weighted cost function 2}
\begin{align}
F_w 
& = \frac{w_r N}{F_r} \log_2\left[ \det\left(\frac{\mathbf{X}^{\ast}_t\mathbf{X}_t^T\varSigma_{\mathbf{G}} +\mathbf{X}^{\ast}_d\mathbf{X}_d^T\varSigma_\mathbf{G} }{\left(\sigma_n^2\right) ^{\frac{L}{N}}}+\mathbf{I}_{N}\right)\right] \\
&\qquad\qquad\qquad+\frac{(1-w_r)L_d}{F_c} \log_2\left[ \det\left(\frac{(\sigma_h^2 - \mathcal{C}^{l}_e)\tilde{\mathbf{H}}\varSigma_{\mathbf{X}_d} {\tilde{\mathbf{H}}}^H}{\frac{P_d}{ L_d}\mathcal{C}^{l}_e+\sigma_n^2}+\mathbf{I}_{N}\right) \right]\label{eq::weighted cost function 2a}\\
& = \frac{w_r N}{F_r}\log_2\left[ \det\left(\frac{\mathbf{\Lambda}_G\left(\mathbf{X}^T_t\mathbf{U}_G \right)^H \mathbf{X}^T_t\mathbf{U}_G+\mathbf{\Lambda}_G\left(\mathbf{X}^T_d\mathbf{U}_G \right)^H \mathbf{X}^T_d\mathbf{U}_G}{\left( \sigma_n^2\right) ^{\frac{L}{N}}}+\mathbf{I}_{N}\right)\right] \notag\\
&\qquad\qquad\qquad+ \frac{(1-w_r)L_d}{F_c} \log_2\left[ \det\left(\frac{(\sigma_h^2 - \mathcal{C}^{l}_e)\mathbf{\Lambda}_{\tilde{\mathbf{H}}}\left(\mathbf{X}_d\mathbf{U}_{\tilde{\mathbf{H}}} \right)^H \mathbf{X}_d\mathbf{U}_{\tilde{\mathbf{H}}}}{\frac{P_d}{ L_d}\mathcal{C}^{l}_e+\sigma_n^2}+\mathbf{I}_{N}\right) \right]\label{eq::weighted cost function 2b}.
\end{align}	
\end{subequations}

Define $\mathbf{\Pi} = (\mathbf{X}^T_d\mathbf{U}_H)^H\mathbf{X}^T_d\mathbf{U}_H = (\mathbf{X}^T_d\mathbf{U}_G)^H\mathbf{X}^T_d\mathbf{U}_G$ with the $(i,j)$-th entry $\varpi_{ij}$, and $\rm{tr}\left(\mathbf{\Pi}\right) = \rm{tr}\left(\varSigma_{\mathbf{X}_d}\right)$.
According to Hadamard's inequality for the determinant and trace of a positive semi-definite Hermitian matrix, we have $\det\left(\mathbf{\Pi}_{N \times N} \right) \leq \prod_{i=1}^{N} \varpi_{ii}$.
We can formulate the MI maximization problem as 
\begin{subequations}\label{eq::op_crb2}
\begin{align}
F_w &\leq \underset{\mathbf{\Pi}}{\max}\sum_{i=1}^{N}\left\lbrace \frac{w_r}{F_r}N\! \log_2\!\left(\!\frac{\lambda_{ii}\left(\frac{P_t}{N}\!+\!\varpi_{ii}\right)}{\left(\sigma_n^2\right) ^{\frac{L}{N}}}\!+\!1\!\right) \!+\!\frac{{1-w_r}}{F_c}L_d\log_2\left(\frac{\left(\sigma_h^2 - \mathcal{C}_e^{l}\right)\mu_{ii}\varpi_{ii}}{\frac{P_d}{L_d}\mathcal{C}_e^{l}\!+\!\sigma_n^2}\!+\!1\!\right)\!\!\right\rbrace\! \label{eq::op_crbla}\\
\quad&\text{subject to} \;\; \Tr\left( \mathbf{\Pi}\right)  \leq P_d;\\
&\qquad\quad\quad\varpi_{ii}\geq 0,\;1\leq i \leq N,\label{eq::op_crblc}
\end{align}
\end{subequations}
where $F_r$ and $F_c$ are the maximum MI~\eqref{eq::MI maximum} in Section~\ref{sec::maximum MI} and the communication capacity~\eqref{eq::op_communication with ce} in Section~\ref{sec::maximum channel capacity}, respectively.

The objective function in \eqref{eq::op_crbla} is concave, since it is a non-negative weighted sum of two concave functions of $\varpi_{ii}$, i.e., $$N\! \log_2\!\left(\!\frac{\lambda_{ii}\left(\frac{P_t}{N}\!+\!\varpi_{ii}\right)}{\left(\sigma_n^2\right) ^{\frac{L}{N}}}\!+\!1\!\right)\!;\;{\rm and}\; L_d\log_2\!\left(\!\frac{\left(\sigma_h^2 \!-\! \mathcal{C}_e^{l}\right)\mu_{ii}\varpi_{ii}}{\frac{P_d}{L_d}\mathcal{C}_e^{l}\!+\!\sigma_n^2}\!+\!1\!\right)\!.$$ 
Besides, the functions $\Tr\left( \mathbf{\Pi}\right) \leq P_d$ and $\varpi_{ii}\geq 0,\;1\leq i \leq N$ are affine.
Therefore, the maximization of the concave problem in \eqref{eq::op_crb2} can be reformulated equivalently to minimize the convex objective. The optimization problem can be solved by using Karush-Kuhn-Tucker (KKT) conditions. Let
\begin{align}
&\nu_i=\frac{\lambda_{ii}}{\left(\sigma_n^2\right)^{\frac{L}{N}}},\ \varphi_i=\frac{\left(\sigma_h^2 - \mathcal{C}_e^{l}\right)\mu_{ii}}{\frac{P_d}{ L_d}\mathcal{C}_e^{l}+\sigma_n^2},
\epsilon=\frac{w_r N}{\ln2 F_r}, \ \text{and}\ \eta=\frac{\left(1-w_r \right)  L_d}{\ln2 F_c},\notag
\end{align}
we have
\begin{subequations}\label{eq::KKT1}
	\begin{align}
	&\zeta - \zeta_i=\frac{\epsilon \nu_i}{1+\nu_i (\frac{P_t}{N}+\varpi_{ii})}+\frac{\eta \varphi_i}{1+\varphi_i \varpi_{ii}};\\
	&\zeta\left(\sum_{i=1}^{N} \varpi_{ii}-P_d\right)=0;\\
	&\zeta_i \varpi_{ii}=0;\\
	&\zeta\geq 0;\; \zeta_i\geq 0,i=1,\cdots,N,
	\end{align}
\end{subequations}
where $\zeta$ and $\zeta_i,\;i=1,\cdots,N$, are the Lagrange multipliers. The optimal solution for~\eqref{eq::KKT1} is given by
\begin{equation}
\begin{aligned}
\hat{\varpi}_{ii} &=\frac{1}{2}\Bigg[\frac{1}{\zeta}\left( \epsilon+\eta\right)-\left(\frac{P_t}{N}+\frac{1}{\nu_i} +\frac{1}{\varphi_i}\right)s+\sqrt{\!\left[\!\left(\frac{P_t}{N}+\frac{1}{\nu_i}-\frac{1}{\varphi_i}\right)+\frac{1}{\zeta}\left(\eta-\epsilon\right) \!\right]^2 +\frac{4\epsilon \eta}{\zeta^2}}\Bigg]^{+},
\end{aligned} 
\end{equation}
where $\left[x\right]^+=\max\left\lbrace x,0\right\rbrace$, and $\hat{\varpi}_{ii},\;i=1,\cdots,N$ satisfy the following equality:
\begin{equation}
\sum_{i=1}^{N} \hat{\varpi}_{ii}-P_d=0.
\end{equation}
The positive $\zeta$ can be obtained by the bisection search over the following interval:
$$0<\frac{1}{\zeta}< \frac{1}{\underset{i}{\min}\left\lbrace \frac{\epsilon \nu_i}{(\frac{P_t}{N}+P_d)\nu_i+1}+\frac{\eta \varphi_i}{\varphi_i P_d+1}\right\rbrace },\; i=1,\cdots,N.$$ 
Once $\zeta$ is obtained, the optimal covariance matrix of the data signals in \eqref{eq::op_crb2} can be obtained. 
Further, we can obtain the maximum relative MI, that is, the sum of the relative communication MI and the relative sensing MI, as given by
\begin{equation}\label{eq::total relative rate}
	\begin{aligned}
	R_{\rm total} = & \sum_{i=1}^{N}\left\lbrace \frac{N}{F_r} \log_2\left(\frac{\lambda_{ii}(\frac{P_t}{N}+\hat{\varpi}_{ii}) }{\left(\sigma_n^2\right)^{\frac{L}{N}}}+1\right)+\frac{{L_d}}{F_c}\log_2\left[\frac{(\sigma_h^2 - \mathcal{C}_e^{l})\mu_{ii}\hat{\varpi}_{ii}}{\frac{P_d}{ L_d}\mathcal{C}_e^{l}+\sigma_n^2}+1\right]\right\rbrace.
	\end{aligned}
\end{equation}

In the case of $w_r=0$, the singular values $\hat{\varpi}_{ii},\,i=1,\cdots,N,$ of the optimal covariance matrix are consistent with \eqref{eq::op power channel capacity} in Section~\ref{sec::maximum channel capacity}. In the case of $w_r=1$, similarly, $\hat{\varpi}_{ii},\,i=1,\cdots,N,$ coincides with~\eqref{eq::op power radar MI}, as described in Section~\ref{sec::maximum MI}.

\section{Simulation Results}\label{sec::simulation}
In this section, we conduct extensive simulations to numerically verify the effectiveness of the proposed methods. A system with $2$ nodes is considered, and each node is equipped with~$8$~antennas. 

We consider (correlated) MIMO Rayleigh fading (complex Gaussian) channels for both communications and sensing, and the channels between them are independent. Both channels remain unchanged during the period of transmitting. Correlated channels are generated based on the Kronecker model, where the normalized correlation matrix has identical diagonal elements and random off-diagonal elements following uniform distributions between 0 and a maximal correlation coefficient $\epsilon_{c}$. The values of $\epsilon_c$ are set as 0.1 and 0.8 for communication and sensing, respectively.
We are particularly interested in the case where communication and sensing channels have the same mean path losses, i.e., $\sigma_h^2=\sigma_g^2=1$. This corresponds to the case where the mean sensing distance is approximately the square root of the communication distance.

 In all the simulations, the noise is complex AWGN. Other simulation parameters are shown in Tab. \ref{table1}, unless stated otherwise. The value of SNR is $(P/L)/\sigma_n^2$. For any given SNR and $L$, we compute the value for $P$ with $\sigma_n^2=1$, and then decide the value for $P_t$ and $P_d$ according to~$\kappa_{\rm{op}}$ in Theorem~\ref{theo::op ene_allo}.

\begin{table}[ht]\small
	\caption{Simulation parameters}
	\begin{center}
		\begin{tabular}{ll}
			\toprule[1.5pt]
			Parameter  & Value \\
			\hline
			Noise power ${\sigma_n^2}$ & 1 \\
			${\sigma_h^2}$ = ${\sigma_g^2}$ & 1 \\
			Number of Antennas & 8 \\
			Number of the training symbols $L_t$ & basic value 8 \\
			Number of the total symbols $L$ & basic values 128 \\
			\toprule[1.5pt]
		\end{tabular}
	\end{center}
	\label{table1}
\end{table}
	
To provide comparable results to the \textit{communication rate}, we introduce the \textit{sensing rate} as the sensing MI per unit time, which can be viewed as the mutual information between the sensing return and the targets. For simplicity, we assume each symbol lasts 1 unit time. Hence, sensing rate, as well as communication rate, equal to the ratio between their respective MI and the number of total transmitted symbols. 
Note that the term, sensing rate, is not widely used in the literature because it implicitly assumes that the sensing channel is not changing during the period of interest. Hence the sensing rate is also related to how fast sensing channel changes. 

For convenience, abbreviations of the waveform design schemes proposed in this paper and other comparison schemes for the legends in figures are listed as follows, where all the schemes include optimal power allocation, except for  ``without power allocation''.
\begin{itemize}
	\item OPTC (or OPTC with CEE): the scheme which optimizes communication only in the presence of CEE at the receiver;	
	\item OPTC without CEE: the scheme which optimizes communication only in the absence of CEE at the receiver;
	\item OPTS: the scheme which optimizes sensing only;
	\item JCAS: the scheme which optimizes both communication and sensing;
	\item Equal: the scheme in which the singular values for data symbol correlation matrix are allocated with equal values;
	\item Random: the scheme in which the singular values for data symbol correlation matrix are allocated with random values;
	\item ``without power allocation'': the scheme in which waveform is optimized for communication only without power allocation between the data and training sequences. 
\end{itemize}
For each result, Monte-Carlo simulations with 5000 independent trials are conducted and the average results are provided.

\begin{figure*}[htbp]
	\begin{minipage}[t]{0.5\textwidth}
		\centering
		\centerline{\includegraphics[width=1\textwidth]{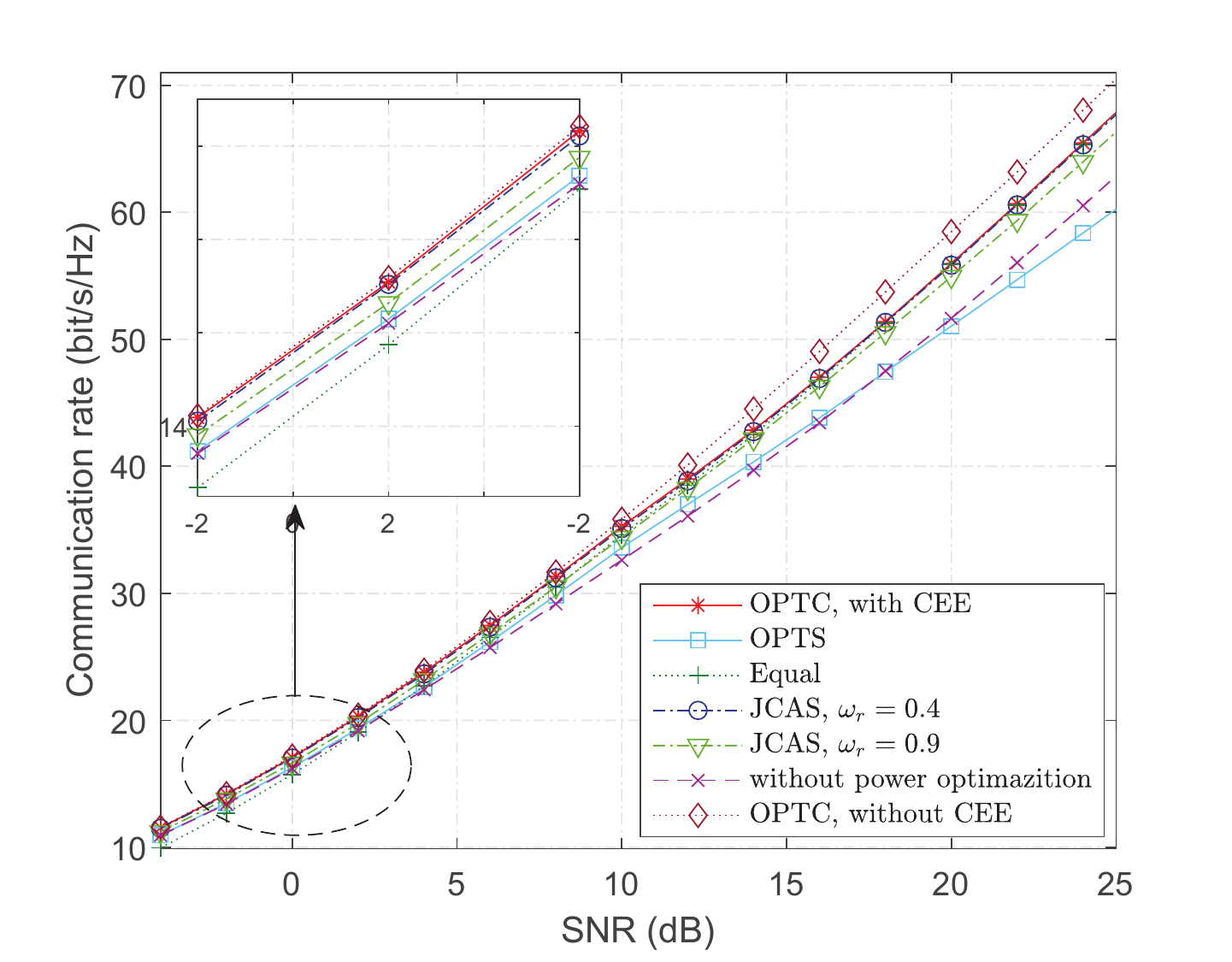}}
		\caption{Communication rate vs. SNR, where $L=128$, $L_t=8$.}
		\label{fig1_COM_rate}
	\end{minipage}
    \hspace{0.2in}
	\begin{minipage}[t]{0.5\textwidth}
		\centering
		\centerline{\includegraphics[width=1\textwidth]{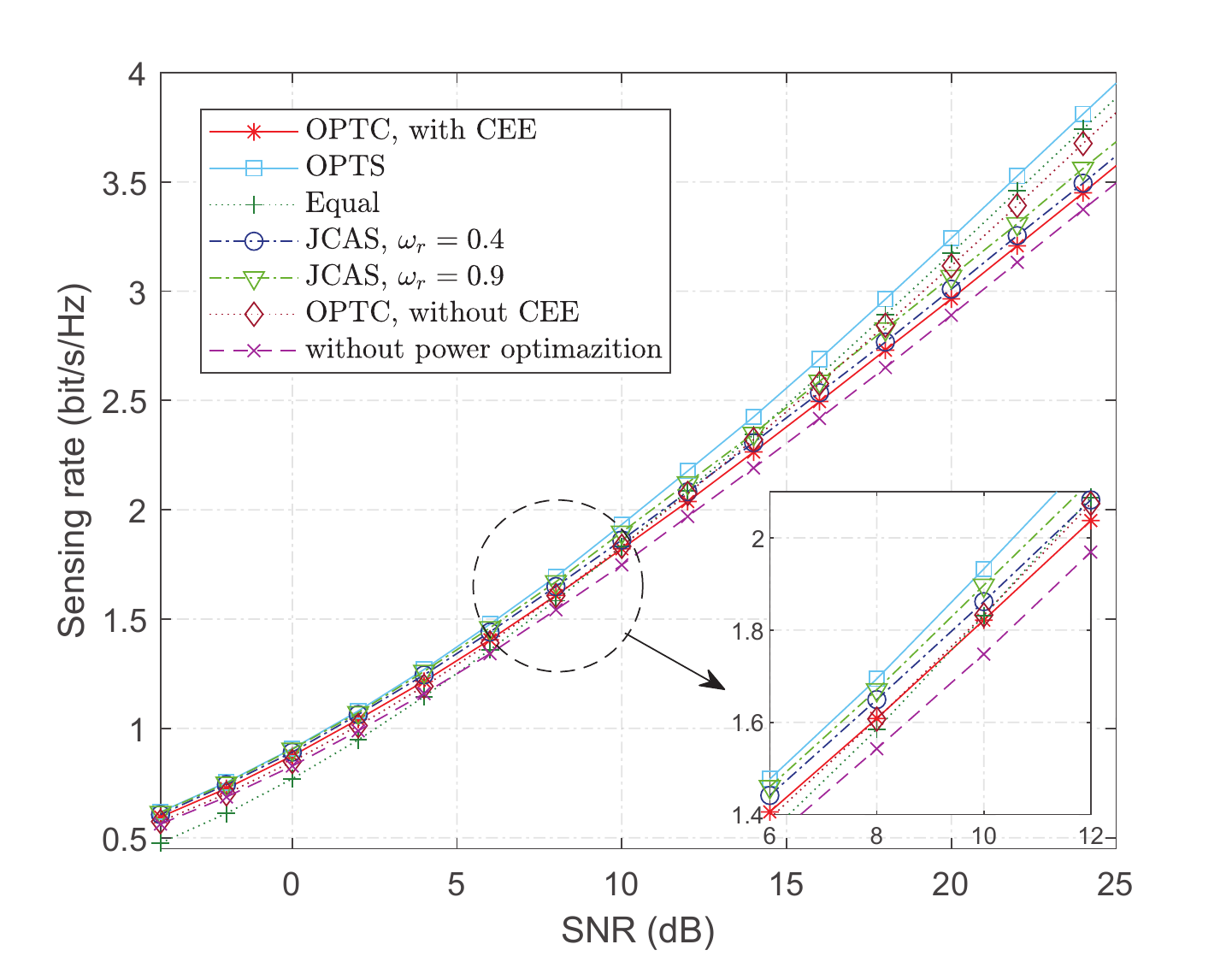}}
		\caption{Sensing rate vs. SNR, where $L=128$, $L_t=8$.}
		\label{fig2_RAR_rate}
	\end{minipage}
\end{figure*}
Fig.~\ref{fig1_COM_rate} plots the communication rate under correlated channels $\mathbf{H}$ with CEE on the SNR for different waveform design schemes. The communication rate without CEE is also plotted for comparison. We can see that the communication rate with CEE is upper bounded by the rate without CEE, and the gap between them decreases with the growth of SNR. The communication rate for OPTC is the maximum in the presence of CEE, as expected. The rate decreases towards that for OPTS with the growth of the weighting factor for sensing. We can also see that the rate for equal power allocation is the lowest in low SNR regimes, and the gap to other schemes decreases with the increase of SNR. In high SNR regimes, the rate approaches the case of OPTC. We also plot the rate without power optimization between the training and data signals. In this case, the rate lies between that for OPTS and Equal in the low SNR regimes, and the gap to that for OPTS decreases with the growth of SNR, and the rate surpasses that for OPTS when the SNR is larger than $16$ dB. This is because the difference between these waveforms decreases gradually as the SNR increases.

In Fig.~\ref{fig2_RAR_rate}, the sensing performance is evaluated under correlated $\mathbf{G}$ for different waveform design schemes. We find that the sensing rate is improved with the increase of the weighting factor for sensing, and approaches the result of OPTS.
For OPTC with CEE, the sensing rate is the lowest of the three in the high SNR region since the communication rate is maximized without taking into consideration the sensing rate. The sensing rate for other schemes are between those for optimal sensing and communications with CEE. 
For OPTC without CEE, the sensing rate is almost the same as, or close to, the rate for other schemes, including OPTC with CEE, OPTC without power optimization, and JCAS, but gradually surpass them. Since the correlated coefficient for $\mathbf{G}$ is larger than $\mathbf{H}$, the larger power allocated for training symbols, the smaller the sensing rate is. Therefore, the rate for OPTC without CEE is larger than those schemes in the high SNR regime. 
The Equal scheme achieves a sensing rate approaching OPTS with the increase of SNR.
The sensing rate for OPTC without power optimization is the lowest among all the schemes in the high SNR regime, and only larger than that for Equal in the lower SNR region. This shows that not only can the optimal power allocation minimize the CEE, but also improve the sensing capability.

\begin{figure*}[htbp]
	\begin{minipage}[t]{0.5\textwidth}
		\centering
		\includegraphics[width=1\textwidth]{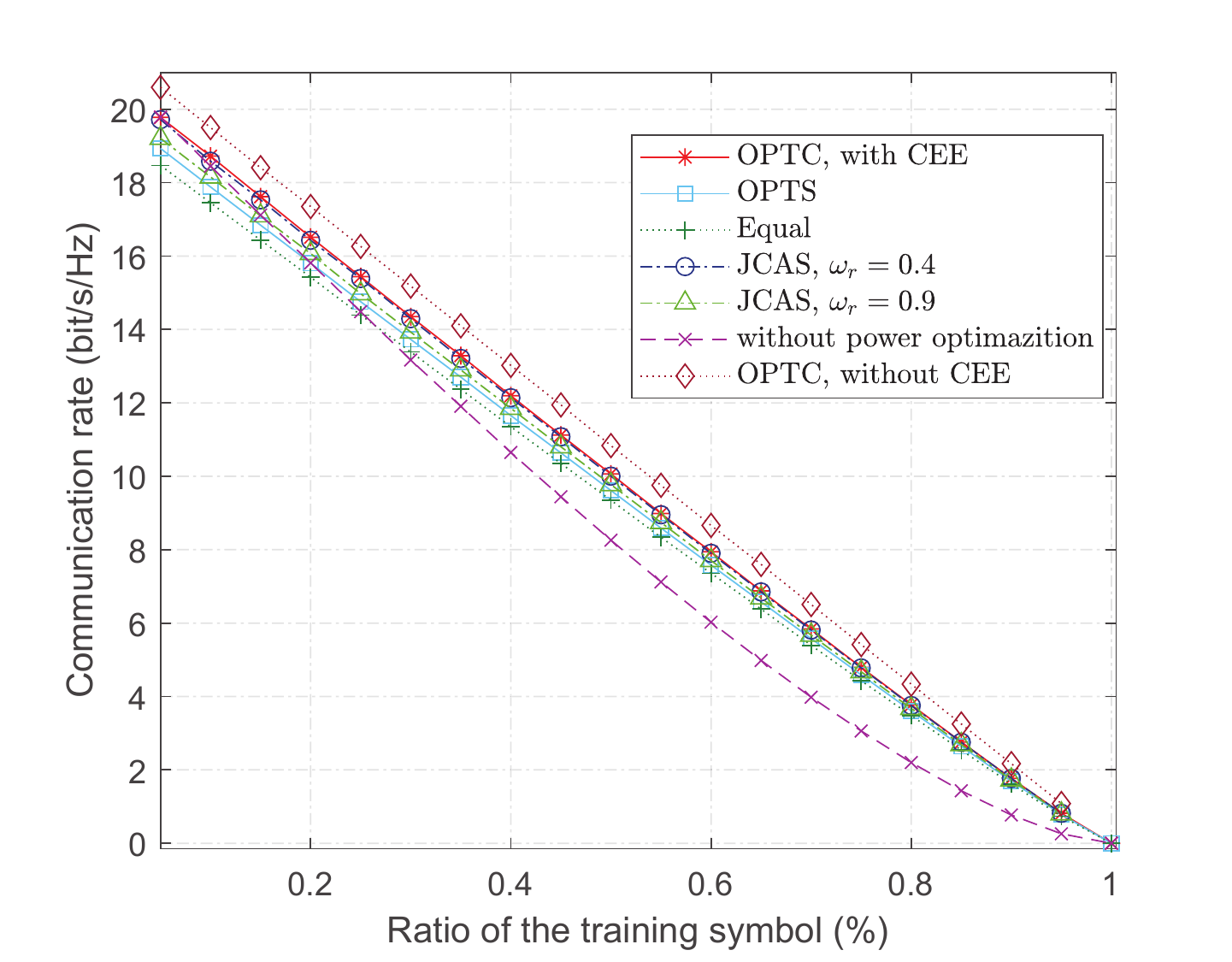}
		\caption{Communication rate vs. the ratio of training symbols where the total number of the transmit symbols $L=160$, ${\rm SNR} = 1$ dB.}
		\label{fig3_CC_rate_vs_ratio}
	\end{minipage}
    \hspace{0.2in}
	\begin{minipage}[t]{0.5\textwidth}
		\centering
		\centerline{\includegraphics[width=1\textwidth]{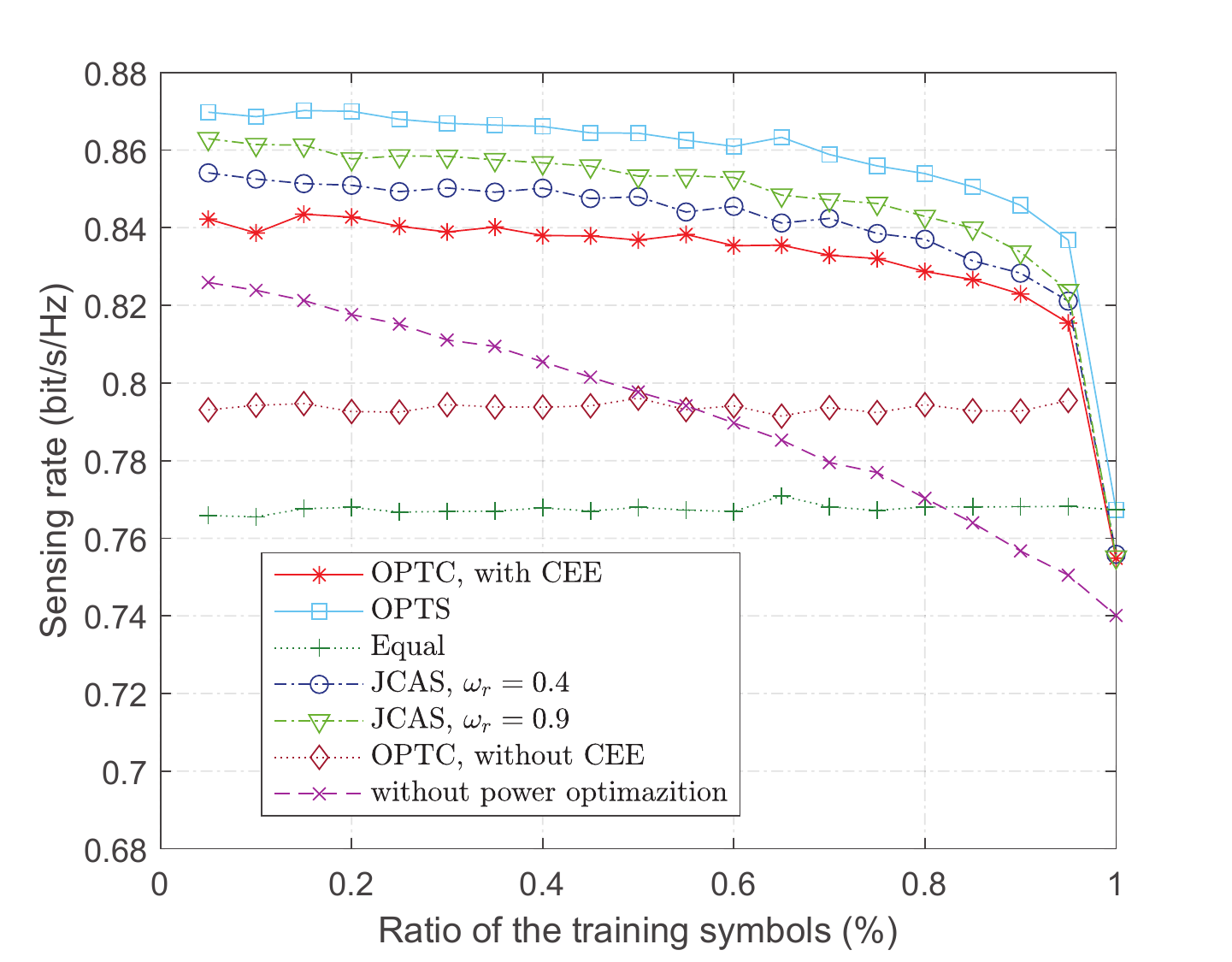}}
		\caption{Sensing rate vs. the ratio of training symbols where the total number of the transmit symbols $L=160$, ${\rm SNR} = 1$ dB.}
		\label{fig3_MI_rate_vs_ratio}
	\end{minipage}
\end{figure*}

Fig.~\ref{fig3_CC_rate_vs_ratio} plots the communication rate against the ratio of the training symbols to the transmit symbols for various waveform design schemes. The different ratio is obtained via changing the length of training and data symbols while keeping the total number unchanged. We see that the communication rate decreases significantly with the growth of the ratio. 
This is because the number of data symbols reduces as the ratio of training symbols increases, resulting in the subsequent decrease in communication rate. 
Moreover, from Fig.~\ref{fig3_CC_rate_vs_ratio}, we see that the communication rate of OPTC without power optimization between the training and data symbols decreases the fastest among all the schemes. In other words, power optimization is important to improve the communication rate. 

Fig.~\ref{fig3_MI_rate_vs_ratio} depicts the sensing rate against the ratio, where $\mathbf{G}$ remains correlated during the period of $L$ symbols. From the figure, we can see that the sensing rate changes little with the increase of the ratio for Equal and OPTC
without CEE.
This is because the orthogonal (training) symbols can achieve the same sensing rate as the (data) symbols with equal power allocation. 
For Equal, both the orthogonal training and data symbols with equal power allocation are used for sensing, and therefore, changes in the ratio have little impact on the sensing rate.  
For OPTC without CEE, all power is used for communication and changes in the ratio do not affect the waveform design for communication, and therefore, the sensing rate of OPTC without CEE keeps unchanged.
 We also see that the curves of OPTC, OPTS, and JCAS show slight declines and approach the curve for Equal as the ratio increases. This is because the impact of the training symbols on the sensing rate is increasingly strong with the growth of the ratio, causing the sensing rate to approach that for Equal (as the training symbols are orthogonal). 
 However, for OPTC without power optimization, the sensing rate declines rapidly with the growth of the ratio. This is because the power allocated to the orthogonal training symbols increases linearly with the ratio. The more power for orthogonal training symbols, the smaller the sensing rate is for the strongly correlated~$\mathbf{G}$.

\begin{figure}[!t]
	\centering 
	\subfigure[Sensing rate \& Communication rate vs. the number of transmit symbols $L$.]{ 
	\centering 
	\label{fig:subfig:a} 
	\includegraphics[width=0.48\textwidth]{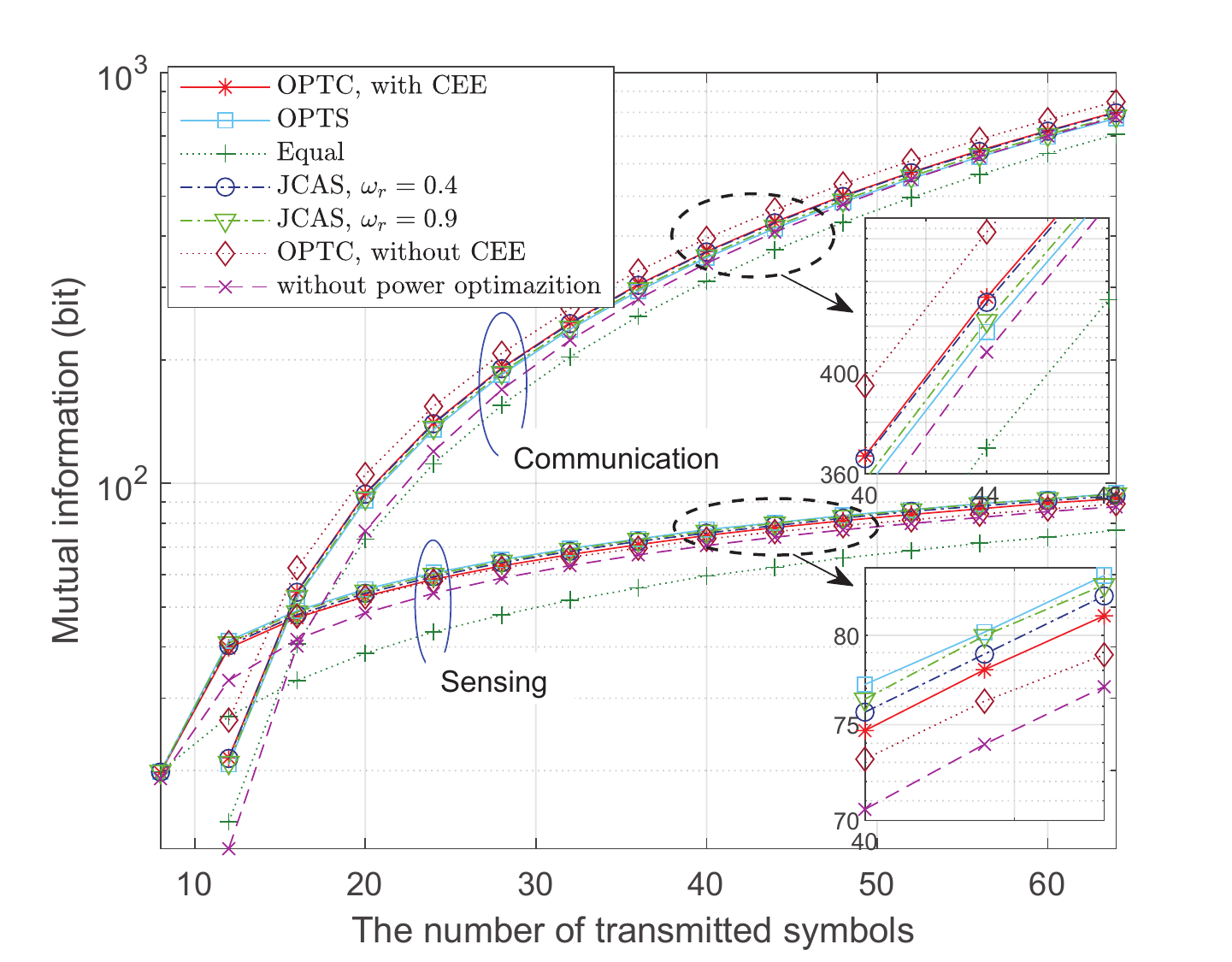}} 
    \subfigure[Mutual information vs. $L$.]{ 
    \centering 
    \label{fig:subfig:b} 
	\includegraphics[width=0.48\textwidth]{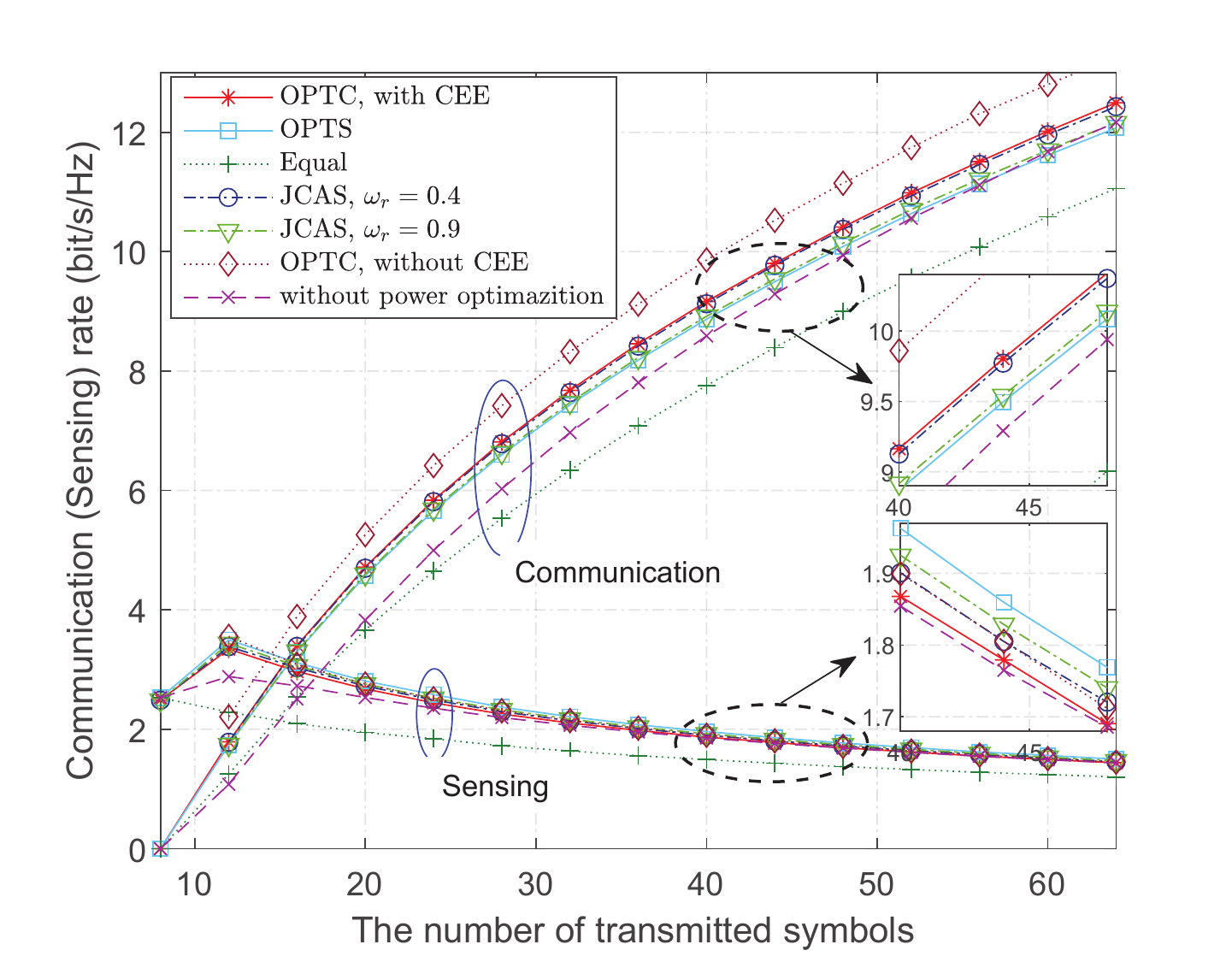}}
    \subfigure[Mutual information vs. $L$.]{ 
    \centering 	
    \label{fig:subfig:c} 
	\includegraphics[width=0.48\textwidth]{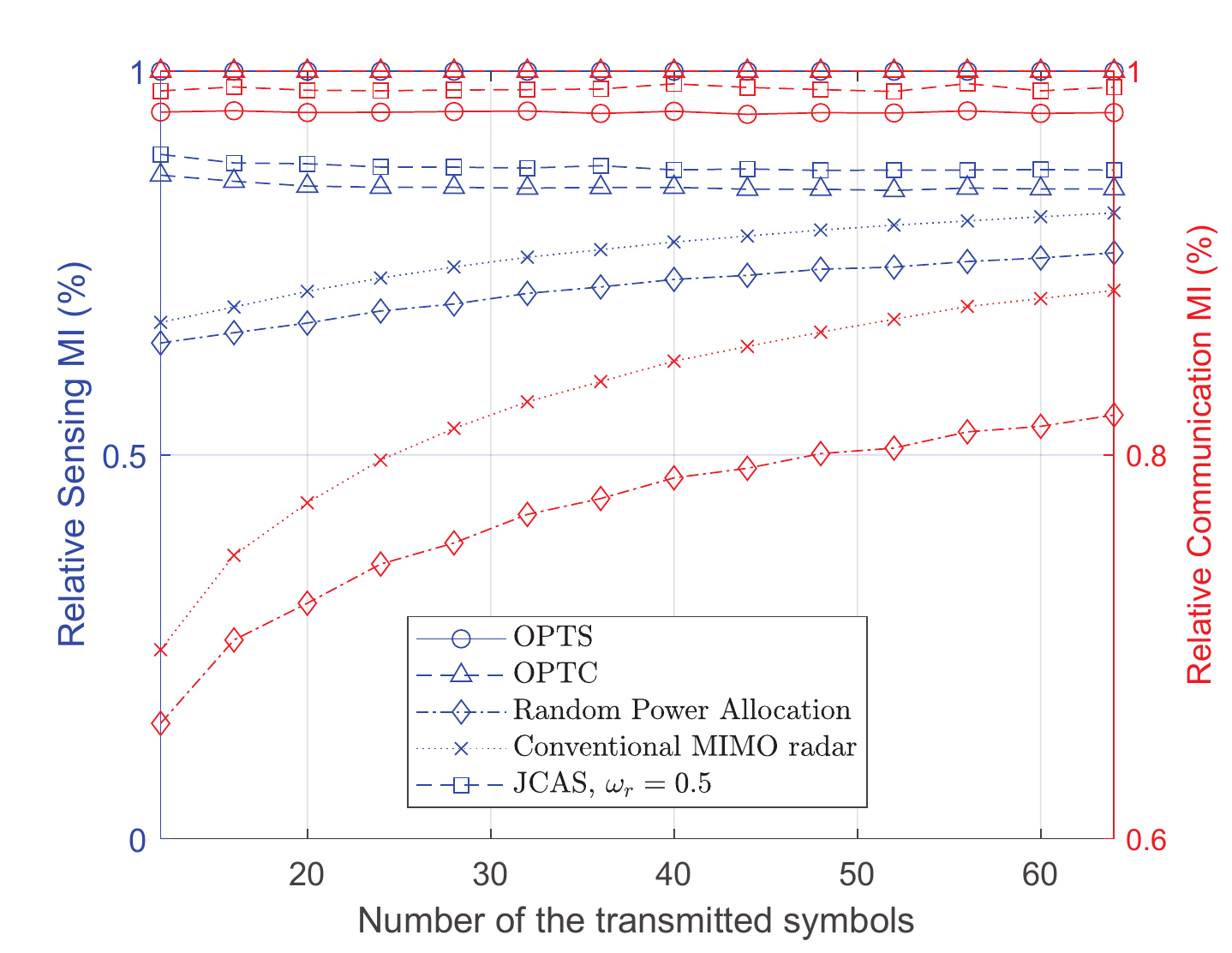}} \caption{Mutual information and rates vs. the total number of transmitted symbols $L$, where $L_t =8$ and ${\rm SNR} = 1$ dB.} \label{fig:MI_CC_vs_L} 
\end{figure}

In the subfigures of Fig.~\ref{fig:MI_CC_vs_L}, both MI and rates for communication and sensing are plotted against the total number of the transmitted symbols $L$. The number of training symbols $L_t$ is fixed to 8, and only the number of data symbols varies. 
In Fig.~\ref{fig:subfig:a}, we can see the MI for both sensing and communication increases with $L$, and the increase rate for communication is much higher than that for sensing. This is because the MI for communication increases almost linearly with the number of data symbols transmitted, but the MI for sensing only increases logarithmically since the sensing channel $\mathbf{G}$ is assumed to remain unchanged and more data symbols only increase the SNR. When $L<16$, the MI for sensing is greater than that for communication. This is because the CEE has a strong impact on the MI for communication.

Fig.~\ref{fig:subfig:b} shows the rates for communication and sensing, corresponding to the MI for them shown in Fig.~\ref{fig:subfig:a}. 
We also find that the communication rate increases with $L$, while the sensing rate increases first and then decreases with $L$. 
This is due to the fact that, as the number of data symbols $L_d$ increases (since $L_t$ is a constant), the communication capacity, i.e., communication rate, increases. On the contrary, as for sensing, with the increase of $L$, the channel variation rate reduces, leading to a decrease in the sensing rate. 

Fig.~\ref{fig:subfig:c} shows the relative sensing and communication MI values, which are normalized to their optimal values, respectively. We can see that there is a large gap between those optimal and random waveform designs for communications. The optimal waveform design for sensing leads to better performance than the random waveform design scheme. 
We also see that the performance gap decreases with the growth of $L$. This is because the signal overhead reduces with data symbols increasing.
For sensing, the optimal waveform design can achieve a better performance than the orthogonal waveform that is employed in conventional MIMO radar~\cite{Song2010Reducing}, supposing the channel correlation is known. Both the optimal sensing and JCAS schemes can lead to better MI. We also find that the relative sensing MI for JCAS slightly decreases first and then converges to a constant value. This is because the ratio $L_t/L$ decreases rapidly and convergences to zero with the increase of $L$. 
\begin{figure*}[htbp]
	\begin{minipage}[t]{0.5\textwidth}
		\centering
		\centerline{\includegraphics[width=1\textwidth]{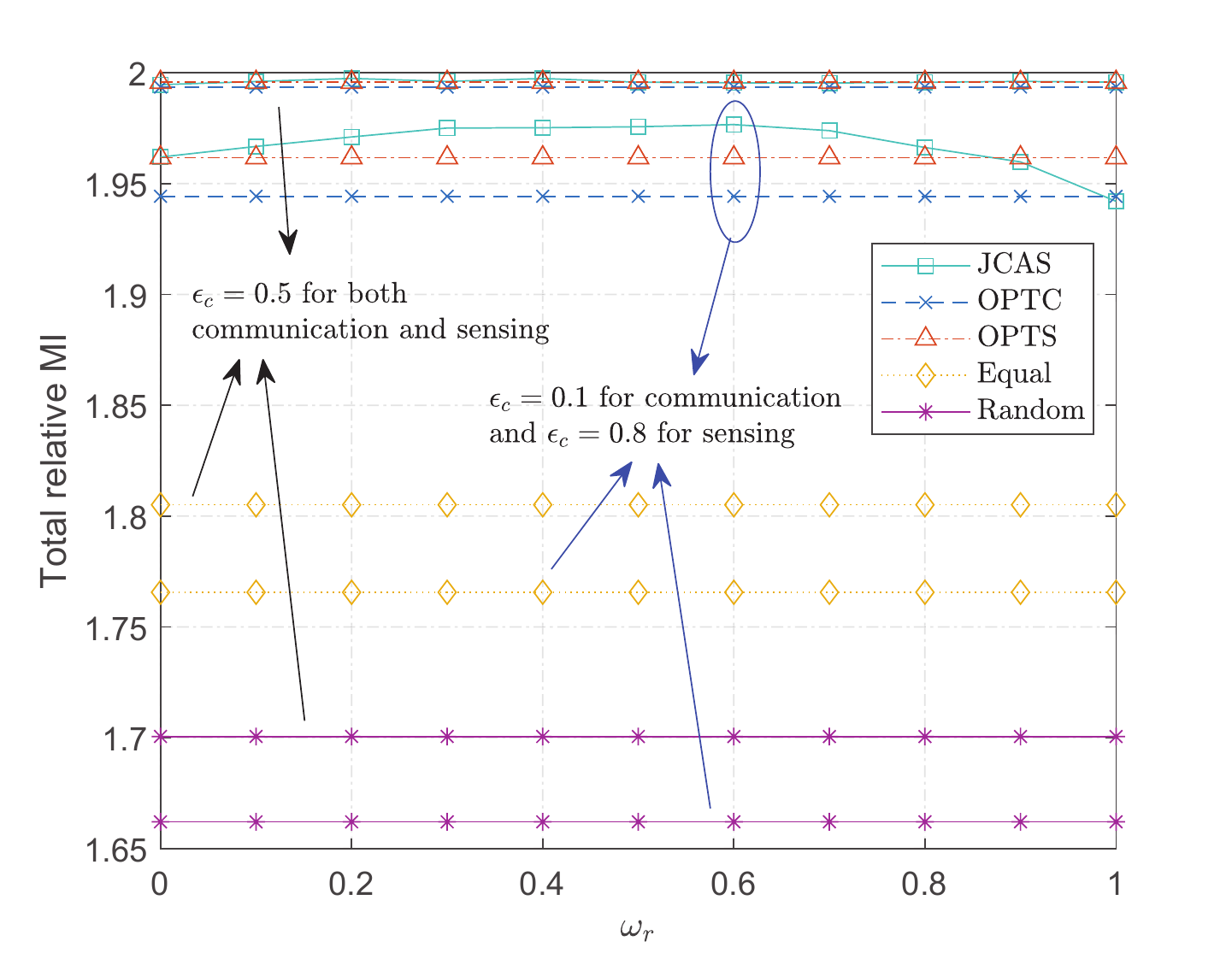}}
		\caption{Total relative MI vs. $\omega_r$ for different waveform design methods with $\text{SNR}=1$ dB, and $\epsilon_{c}$ for communication and sensing are 0.8 and 0.1, respectively.}
		\label{fig_7_ralative_rate_vs_w}
	\end{minipage}
    \hspace{0.2in}
	\begin{minipage}[t]{0.5\textwidth}
		\centering
		\centerline{\includegraphics[width=1\textwidth]{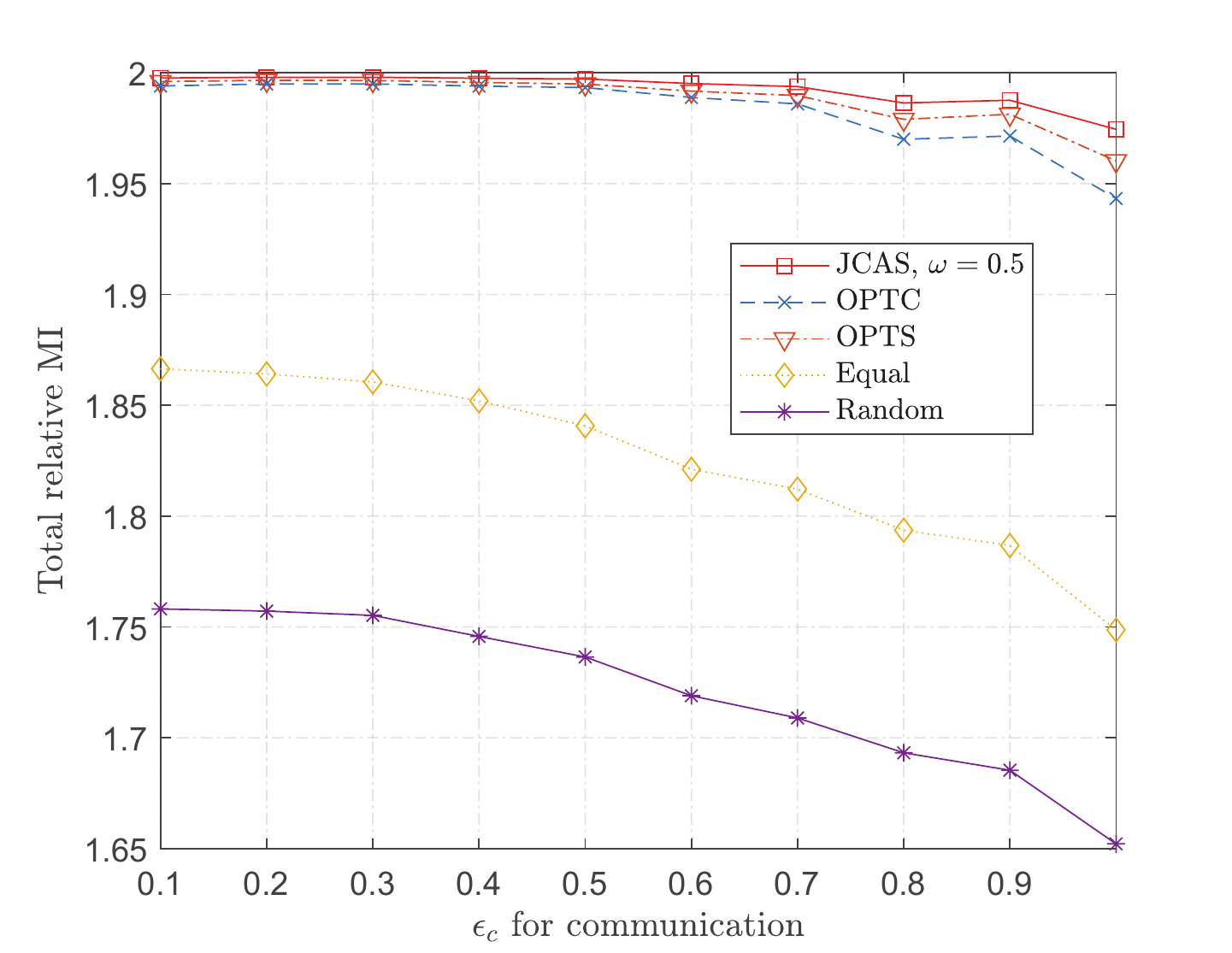}}
		\caption{Total relative MI vs. the maximal correlation coefficient~$\epsilon_{c}$ for communication channels, where $\text{SNR}=1$ dB, and $\epsilon_{c}$ for sensing is 0.3.}
		\label{fig_8_ralative_rate_vs_cor}
	\end{minipage}
\end{figure*}

Fig.~\ref{fig_7_ralative_rate_vs_w} plots the sum of the relative communication and sensing MI against the weighting factor for sensing under different waveform design schemes and correlation coefficient $\epsilon_{c}$. We can see that the total relative MI for JCAS is always the highest in the case of $\epsilon_{c} = 0.1$ for communication and $\epsilon_{c} =0.8$ for sensing. The total relative MI for the JCAS increases first and then decreases with the growth of $w_r$, while the total relative MI value for other schemes does not change with $w_r$. The total relative MI for the equal power allocation waveform design is the lowest. 
As expected, the total relative MI values under JCAS with $w_r=0$ and $w_r=1$ are equal to those for OPTC and OPTS, respectively. 
In the case of $\epsilon_{c} = 0.5$ for both communication and sensing, we find the weighted sum of the communication and sensing MI are almost the same across OPTC, OPTS and JCAS schemes. In other words, a single waveform can be optimized to maximize both communication and sensing MI when the communication and sensing channels have the same correlation characteristics.

Fig.~\ref{fig_8_ralative_rate_vs_cor} plots the sum of the relative communication and sensing MI against the maximum correlation coefficient $\epsilon_c$ for sensing channel $\mathbf{G}$. We can see that the total relative MI decreases with the growth of the correlation coefficient, especially when the coefficient is large. We also find that the JCAS scheme is less affected by the channel correlation and outperforms the other schemes. On the contrary, the increase in the correlation coefficient has a significant impact on the total relative MI for random and equal power allocation schemes, and the MI is drastically reduced with the correlation increases. 

\begin{figure}
	\centering
	\includegraphics[width=0.5\textwidth]{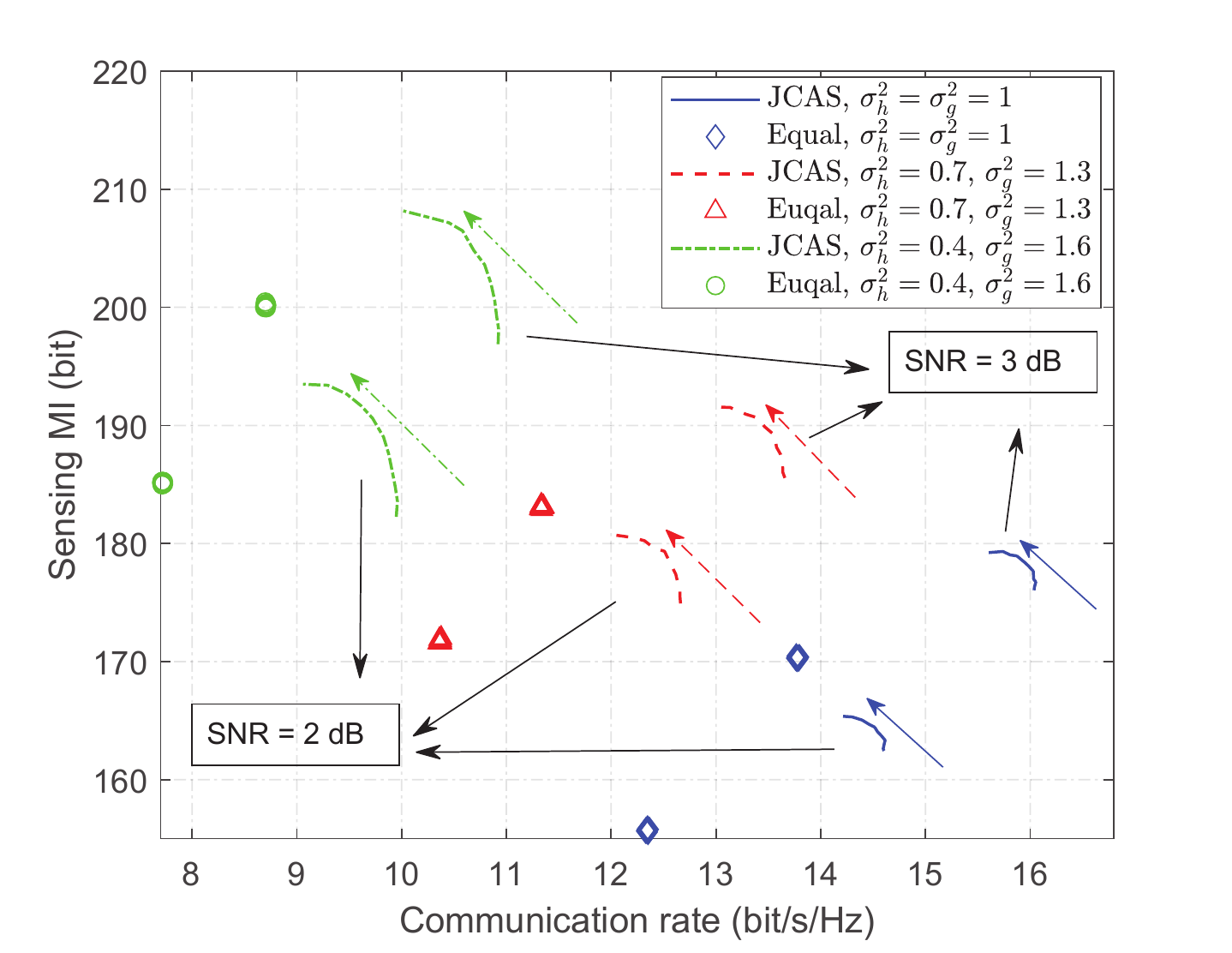}
	\caption{Trade-off curve between communication and radio sensing, where $\epsilon_c$ for communication and sensing are 0.8 and 0.3, respectively.}
	\label{fig_JCAS_pareto}
\end{figure}
Fig.~\ref{fig_JCAS_pareto} demonstrates how the communication rate and sensing MI change with the weighting factor $\omega_r$ which increases from $0$ to $1$ along the direction of the arrow. We consider two cases where the mean path losses for sensing and communication are (i) the same, i.e., $\sigma_h^2=\sigma_g^2=1$, and (ii) different, i.e., $\sigma_h^2=0.7$ and $\sigma_g^2=1.3$, or $\sigma_h^2=0.4$ and $\sigma_g^2=1.6$.  
In both cases, we can see that the sensing MI improves with the growth of $w_r$, and the sensing MI and communication rate are enhanced with the increase of SNR. We also see that the JCAS waveform scheme outperforms ``Equal''. 
We also find that there is a gap between the trade-off curves for the two cases, which is due to the different path losses for communication and sensing.
In addition, with the decrease of the mean path loss for communication, the variation range of the sensing MI increases. In other words, the stronger path loss for communication is, or the weaker path loss for sensing is, the greater the sensing MI is.
In practice, the appropriate weighting value can be selected based on the trade-off curve to meet the requirements of both communication and radio sensing.

\section{Conclusion}\label{sec::conclusion}
We presented the optimal waveform design methods based on MI for MIMO JCAS systems by considering a typical packet structure, including training and data symbols. We proposed an optimal power allocation scheme under MMSE estimators for MIMO communication channels. Among the three optimization strategies we studied, the design that maximizes the weighted sum of relative MI is shown to achieve the best overall performance for joint sensing and communication. The design is also less affected by varying channel correlation than the other two waveform design methods. The methodologies presented in this paper can be further extended to study waveform optimization in JCAS MIMO uplink and multiuser-MIMO systems.

Other important observations obtained in this paper are summarized as follows:
\begin{itemize} 
	\item In most cases, the signal waveform cannot be optimized to maximize both the communication and sensing MI at the same time. However, the JCAS waveform design is nearly optimal for both in the low SNR regime where the noise has a dominant impact on the design.
	\item The proposed optimal power allocation can efficiently improve the MI for communication and has an insignificant impact on the sensing MI. When there are more data symbols than training symbols, the sensing MI under power allocation is higher than the case without power allocation. 
	\item The ratio of mean path losses between communication and sensing can have a strong impact on the range of optimized MI values. A higher ratio can lead to a larger range. If the ratio is small enough, our waveform design can approach the optimal MI values for both communication and sensing. 
\end{itemize}

\appendices

\section{Proof of Theorem~\ref{theo::op ene_allo}}\label{proof::theo 1}
Let the SVD of the spatial correlation matrix of $\tilde{\mathbf{H}}$ be
\begin{equation}\label{eq::Sigma_H tilde}
\varSigma_{\tilde{\mathbf{H}}}=\frac{1}{N}\mathbb{E}\{\tilde{\mathbf{H}}\tilde{\mathbf{H}}^H\}=\mathbf{U}_{\tilde{\mathbf{H}}}\mathbf{\Lambda}_{\tilde{\mathbf{H}}}\mathbf{U}_{\tilde{\mathbf{H}}}^H,
\end{equation}
where $\mathbf{U}_{\tilde{\mathbf{H}}}$ is a unitary matrix and 
$\mathbf{\Lambda}_{\tilde{\mathbf{H}}}={\diag}\left\{ \mu_{11},\cdots,\mu_{ii},\cdots,\mu_{NN}\right\}$
is a diagonal matrix with~$\mu_{ii}$ being the singular values.	
Based on \eqref{eq::MI communication} and the lower bound of CEE, we can get an upper bound for the MI between $\mathbf{X}_d$ and ${\mathbf{Y}}^d_{\rm com}$ as 
\begin{subequations}\label{eq::MI communication CEE1}
	\begin{align}
	& I\left( \mathbf{X}_d;{\mathbf{Y}}^d_{\rm com}|\hat{\mathbf{H}}\right) \\
	&\!\leq\! L_d \log_2\!\left[\! \det\!\left(\!\frac{\left(\sigma_h^2 \!-\! \mathcal{C}_e\right)\mathbf{X}_d\mathbf{U}_{\tilde{\mathbf{H}}}\mathbf{\Lambda}_{\tilde{\mathbf{H}}}\mathbf{U}_{\tilde{\mathbf{H}}}^H\mathbf{X}_d^H}{\frac{P_d}{L_d}\mathcal{C}_e+\sigma_{n}^2}\!+\!\mathbf{I}_{N}\!\right) \!\right]\!\label{eq::MI communication CEE1 a}\\
	&\!=\! L_d \log_2\!\left[\!\det\!\left(\!\frac{\left(\sigma_h^2 \!-\! \mathcal{C}_e\right)\mathbf{\Lambda}_{\tilde{\mathbf{H}}}\left(\mathbf{X}_d\mathbf{U}_{\tilde{\mathbf{H}}} \right)^H \mathbf{X}_d\mathbf{U}_{\tilde{\mathbf{H}}} }{\frac{P_d}{ L_d}\mathcal{C}_e+\sigma_{n}^2}\!+\!\mathbf{I}_{N}\!\right) \!\right]\!, \label{eq::MI communication CEE1 b}
	\end{align}
\end{subequations}
where the equality in~\eqref{eq::MI communication CEE1 a} can be achieved when the lower bound CEE $\mathcal{C}_e$ is achieved;~\eqref{eq::MI communication CEE1 b} is based on the \textit{Sylvester's determinant} theorem~\cite{gilbert1991positive}.

Let $\mathbf{\Xi}=\left(\mathbf{X}_d\mathbf{U}_{\tilde{\mathbf{H}}} \right)^H \mathbf{X}_d\mathbf{U}_{\tilde{\mathbf{H}}}=(\mathbf{Y}^d_{\rm com})^H\mathbf{Y}^d_{\rm com}$, and its $(i,j)$-th entry is $\xi_{ij}$.
Based on Hadamard's inequality for the determinant and trace of a positive semi-definite Hermitian matrix, we have $\det\left(\mathbf{\Xi}_{N \times N} \right) \leq \prod_{i=1}^{N} \xi_{ii}$.
The upper bound of the MI between $\mathbf{X}_d$ and ${\mathbf{Y}}^d_{\rm com}$ can be obtained as 
\begin{equation}\label{eq::MI communication CEE2}
\begin{aligned}
I\left( \mathbf{X}_d;{\mathbf{Y}}^d_{\rm com}|\hat{\mathbf{H}}\right) & \!\leq\! \sum_{i=1}^{N} \!\underbrace{\left\lbrace L_d\, \log_2\!\left[\!\frac{\left(\sigma_h^2 \!-\!  \mathcal{C}_e\right)\mu_{ii}\xi_{ii} }{\frac{P_d}{ L_d}\mathcal{C}_e+\sigma_n^2}\!+\!1\!\right]\right\rbrace }_ {f\left(\xi_{ii}\right)},
\end{aligned}
\end{equation}
where the equality is achieved if and only if $\mathbf{\Xi}$ is a diagonal matrix.

It is easy to see that $f\left(\xi_{ii}\right)$ is a monotonically decreasing concave function of $\xi_{ii}$. Based on Jensen's inequality, we can obtain the expectation of $I\left( \mathbf{X}_d;{\mathbf{Y}}^d_{\rm com}|
\hat{\mathbf{H}}\right)$, i.e., $\mathbb{E}\left[I\left( \mathbf{X}_d;{\mathbf{Y}}^d_{\rm com}|\hat{\mathbf{H}}\right) \right]$, as given in \eqref{eq::average MI commmunication}, where \eqref{eq::average MI commmunication c} is obtained by substituting {$\mathbb{E}\left[ \xi_{ii}\right]= \frac{1}{N} \Tr\left(\mathbf{\varSigma_{\mathbf{X}_d}}\right)=  \frac{1}{N} \Tr\left(\mathbf{\Xi}\right) = L_d \sigma_d^2=\frac{\kappa P}{N}$} and ${(1-\kappa)P} =N L_t \sigma_t^2$ into \eqref{eq::average MI commmunication b};
\begin{subequations}\label{eq::average MI commmunication}
	\begin{align}
	\mathbb{E}\left[I\left( \mathbf{X}_d;{\mathbf{Y}}^d_{\rm com}|\hat{\mathbf{H}}\right) \right] & \leq \mathbb{E}\left[\sum_{i=1}^{N} f\left(\xi_{ii}\right)\right] \\
	& \!=\! \sum_{i=1}^{N} \mathbb{E}\left[ f\left(\xi_{ii}\right)\right]
	\leq \sum_{i=1}^{N} f\left(\mathbb{E}\left[ \xi_{ii}\right]\right)\\
	& \!=\! L_d\sum_{i=1}^{N} \left\lbrace\log_2\left[\frac{\left(\sigma_h^2- \mathcal{C}_e\right)\mathbb{E}\left[\xi_{ii}\right] }{\frac{P_d}{ L_d}\mathcal{C}_e+\sigma_n^2}+1\right] \right\rbrace \label{eq::average MI commmunication b}\\
	& \!=\! L_d\sum_{i=1}^{N} \!\left\lbrace\!\log_2\!\left[\!\frac{P L_d}{(L_d\!-\!N)N\sigma_n^2} \frac{\kappa\left(1-\kappa\right)}{\!-\kappa \!+\!\frac{ L_d}{L_d-N}\left(1\!+\!\frac{N \sigma_n^2}{P\sigma_h^2}\!\right) }\!+\!1\!\right]\! \!\right\rbrace\! \label{eq::average MI commmunication c}.
	\end{align}
\end{subequations}
From~\eqref{eq::average MI commmunication c}, we can see that different values of~$\kappa$ can lead to different mean MI values for a given total energy $P$ and the number of antennas $N$ through the SNR, denoted by 
\begin{equation}
\rho = \frac{L_d P}{(L_d-N)N\sigma_n^2}\cdot \frac{\kappa\left(1-\kappa\right)}{\!-\kappa \!+\!\frac{L_d}{L_d-N}\left(1+\frac{N \sigma_n^2}{P\sigma_h^2}\right) }.
\end{equation}
Referring to the cases considered in~\cite{Hassibi2003How}, to maximize $\rho$ over $0\leq \kappa \leq 1$, we can separately consider the following three cases:
\begin{enumerate}
	\item $L_d = N$: The maximal $\rho$, denoted by $\rho_{\max}$, is obtained as
	$$\rho_{\max} = \frac{P^2 \sigma_h^4}{4N\sigma_n^2(N\sigma_n^2+P\sigma_h^2)},$$  
	from which it follows that $\kappa_{\rm{op}} = \frac{1}{2}$.
	\item $L_d > N$: We rewrite $\rho$ as
	$$\rho=\frac{L_d P}{(L_d-N)N\sigma_n^2}\cdot \frac{\kappa
		\left(1-\kappa\right)}{-\kappa + \Gamma},$$
	where $\Gamma = \frac{L_d}{L_d-N}\left(1+\frac{N \sigma_n^2}{P\sigma_h^2}\right)>1$.
	The maximal SNR $\rho_{\max}$ can be obtained as
	$$\rho_{\max} \!=\! \frac{L_d P}{(L_d\!-\!N)N\sigma_n^2}{\!\left(\! \sqrt{\Gamma}-\sqrt{\!\Gamma\!-1}\right) ^2},$$
	and it follows that $\kappa_{\rm{op}} = \Gamma-\sqrt{\Gamma(\Gamma-1)}$.
	\item $L_d < N$: We rewrite $\rho$ as
	$$\rho=\frac{L_d P}{(N- L_d)N\sigma_n^2}\cdot \frac{\kappa\left(1-\kappa\right)}{\kappa - \Gamma},$$ 
	where $\Gamma = \frac{L_d}{L_d-N}\left(1+\frac{N \sigma_n^2}{P\sigma_h^2}\right) < 0$.
	The maximal SNR $\rho_{\max}$ can be obtained as
	$$\rho_{\max} \!=\! \frac{L_d P}{(L_d\!-\!N)N\sigma_n^2}{\!\left(\! \sqrt{\!-\Gamma}-\!\sqrt{\!-\Gamma\!-1}\right) ^2},$$
	and it follows that $\kappa_{\rm{op}} = \Gamma+\sqrt{\Gamma(\Gamma-1)}$.
\end{enumerate}

Therefore, we can obtain the lower bound of the CEE as presented in~\eqref{eq::CRLB}.



\ifCLASSOPTIONcaptionsoff
  \newpage
\fi


\bibliographystyle{IEEEtran}
\bibliography{reference_sensing}

\end{document}